\documentclass[11pt,reqno]{article}
\usepackage{amsmath}
\usepackage{amsmath}
\usepackage{amsthm}
\usepackage{amssymb}
\usepackage{epsfig}

\usepackage{calc}
\usepackage{xspace}
\usepackage{comment}
\usepackage{paralist}

\usepackage{url}
\usepackage{setspace}
\usepackage{rotating}
\usepackage{amsmath,amssymb}
\usepackage{subfigure}
\usepackage{array,arydshln}	
\usepackage{multirow}
\usepackage{booktabs}
\usepackage{color,colortbl}
\usepackage{authblk}
\usepackage{hyperref}
\usepackage{enumitem}
\usepackage{tabularx}	
\usepackage[table]{xcolor}
\usepackage{adjustbox}
\usepackage{float}
\usepackage{makecell}
\usepackage{cite}

\usepackage{natbib}
\DeclareMathOperator*{\argmin}{argmin}
\DeclareMathOperator*{\argmax}{argmax}

\newtheorem{prop}{Proposition}[section]

\newtheorem{thm}{Theorem}
\newtheorem{lm}{Lemma}

\newtheorem{@remark}{\bf Remark}
\newenvironment{remark}{\begin{@remark}\rm}{\end{@remark}}

\newcommand {\A} {\alpha}
\newcommand {\B} {\beta}
\newcommand{\W}{\omega}
\newcommand {\E} {\mathrm{E}}
\newcommand {\T} {\theta}
\newcommand{\HT} {\hat{\theta}_{\A,n}}
\newcommand {\ep} {\epsilon}
\newcommand {\pa} {\partial}
\newcommand {\paa}{\partial^2_{\T\T'}}

\newcommand {\ud}{\mathrm{d}}

\title{ Sequential change point test in the presence of outliers\\
: the density power divergence based
approach}
\author{Junmo Song}
\affil{Department of Statistics, Kyungpook National University}
\date{}

\begin{document}
\maketitle
\begin{abstract}
In this study, we consider a problem of monitoring parameter changes particularly in the presence of outliers.
To propose a sequential procedure that is robust against outliers, we use the density power divergence to derive a detector and stopping time that make up our procedure. We first investigate the asymptotic properties of our sequential procedure for i.i.d. sequences  and then extend the proposed procedure to stationary time series models, where we provide a set of sufficient conditions under which the proposed procedure has an asymptotically controlled size and consistency in power. As an application, our procedure is applied to the GARCH models.
We demonstrate the validity and robustness of the proposed procedure  through a simulation study. Finally, two real data analyses are provided to illustrate the usefulness of the proposed sequential procedure.
\\
\end{abstract}

\section{Introduction}\label{Sec:intro}

We often face various events that can cause structural changes in underlying dynamics. In the field of finance, for example, global financial crises or changes of monetary policy can be such events.  This changes are usually represented through structural breaks or parameter changes in a fitted model.  The statistical analysis for testing or detecting such changes is referred to as change point analysis. Due to its importance in statistical inferences and actual practice,  vast amount of papers have been devoted to this area.  For historical background and general review, see  \cite{csorgo:horvath:1997}, \cite{aue:horvath:2013}, \cite{horvath:rice:2014} and the references therein.

Most of the literature have dealt with the change point problem in the retrospective settings, that is, testing for structural or parameter changes in an observed (or historical) data set. However, since \cite{chu.et.al:1996} developed the sequential  test procedure that performs a monitoring of parameter change using newly arrived data, this type of test procedure has also attracted special attention from many authors.  See, for example, \cite{leisch.et.al:2000}, \cite{horvath.et.al:2004}, \cite{berkes.et.al:2004seq}, \cite{aue.et.al:2006},  \cite{aue2014},  \cite{bardet:kengne:2014},  and \cite{kirch2018}.

In this study, we are interested in a sequential test for parameter change, particularly in the presence of outliers.
Indeed, our study is started from the empirical observation that parameter changes  and  deviating observations can coexist and, in this case, the existing parameter change tests are likely to draw an erroneous conclusion. Strictly speaking, when atypical observations are included in a data set being suspected of having parameter changes, whether the testing results are due to genuine changes or not cannot be readily determined. Such concern has been addressed, for example, in the recent studies by \cite{song:kang:2019} and \cite{song:2020}, where they reported that naive parameter change tests can be severely distorted by outliers.


The problem  still applies to the case of the sequential test procedures. However, it should be more emphasized that the impact of outliers in the sequential settings may be more complex than in the retrospective settings because the influence of outliers can vary depending on the location of outliers.  Here, we note that outliers can be included in a historical data or they can occur in arrived data and that both cases are also possible. In particular, the situation may be more difficult when there is no outlier in the historical data but outliers occur during the monitoring period. In such cases, the analysis is likely to be performed through a non-robust method and, of course, it is very possible that outliers occurring in arrived data mislead the sequential analysis.
 In this regard, there is more need to develop a robust sequential test procedure, but little effort has been made so far.

 In the literature, robust parameter change tests were dealt with by several authors. \cite{tsay:1988} investigated a procedure for detecting outliers, level shifts, and variance change in a univariate time series and \cite{lee:na:2005} and \cite{kang:song:2015} introduced a estimates-based CUSUM test using a robust estimator. Recently, \cite{fearnhead:rigaill:2019} proposed a robust penalized cost function for detecting changes in the location parameter. \cite{song:kang:2019} introduced robust tests based on the divergence introduced below and  \cite{song:2020} proposed a trimmed residual based robust test. These works, however, were conducted in the retrospective settings. 

The aim of this study is to propose a robust sequential test procedure for parameter change. For this, we use the density power (DP) divergence to construct a detector that gives an alarm signal for indicating a parameter change in monitoring period. Since \cite{BHHJ:1998} introduced DP divergence (DPD), the divergence has been successfully used in developing robust estimators. The main property of the DP divergence is that it provides a smooth bridge between Kullback-Leibler (KL) divergence and $L_2$-distance. So, the estimators induced from the DP divergence, the so-called minimum DPD estimator (MDPDE), get efficient and robust properties. For more details, see \cite{BHHJ:1998} and \cite{fujisawa:eguchi:2008}. The DPD based tests were also introduced by several authors. \cite{basu.et.al:2013, basu.et.al:2016} used the objective function of the MDPDE to propose a Wald-type test and \cite{song:kang:2019}  also introduced a DPD based score type test for parameter change, which has recently been extended to integer-valued models and dynamic factor models.  See  \cite{kang:song:2020}, \cite{kim:lee:2020}, and \cite{kim.et.al:2021}.  Like MDPDE, these DPD based tests were also found to inherit the robust and efficient properties from the DP divergence, thereby motivating us to consider a DPD based sequential test.

This paper is organized as follows. In Section \ref{Sec:2}, we develop a DPD-based  sequential test for parameter change in i.i.d. sequences and investigate its asymptotic behaviours.  In Section \ref{Sec:3}, we extend our method to general time series models and provide an application to GARCH models. We examine our test procedure  numerically through Monte Carlo simulations in Section \ref{Sec:sim}. Section \ref{Sec:real} illustrates  two real data applications and Section \ref{Sec:con} concludes the paper. The technical proofs are provided in Appendix.

\section{DP divergence based sequential change point test}\label{Sec:2}

We first review the MDPDE introduced by \cite{BHHJ:1998} and provide some conditions for the strong consistency of the estimator. Then, we  propose a detector and  a stopping time for sequential test procedure based on the DP divergence.

For two density functions $f$ and $g$, \cite{BHHJ:1998} defined the DP divergence, $d_\A(f,g)$, as follows:
\begin{eqnarray*}\label{DPD}
d_\alpha (g, f):=\left\{\begin{array}{lc}
\displaystyle\int\Big\{f^{1+\alpha}(x)-(1+\frac{1}{\alpha})\,g(x)\,f^\alpha(x)+\frac{1}{\alpha}\,g^{1+\alpha}(x)\Big\} dx &,\alpha>0, \vspace{0.3cm}\\
\displaystyle\int g(x)\big\{ \log g(x)-\log f(x) \big\} dx
&,\alpha=0.
\end{array} \right.
\end{eqnarray*}
Here, note that the divergence becomes KL divergence and $L_2$ distance when $\A=0$ and $\A=1$, respectively.
Since $d_\alpha(f,g)$ converges to $d_0(f,g)$ as $\A\rightarrow 0$, the DP divergence with
$0<\alpha<1$ provides a smooth bridge between KL divergence and $L_2$ distance. 

Let $X_1, \cdots, X_n$ be a random sample from an unknown density $g$ and consider a family of parametric densities $\{ f_{\theta} | \theta \in \Theta \}$. Then, the MDPDE with respect to the parametric family  $\{ f_{\theta} \}$ is defined as the minimizer of  the empirical version of the divergence $d_\A(g,f_\T)$. That is,
\begin{eqnarray}\label{MDPDE}
\hat \theta_{\alpha, n} = \argmin_{\theta \in \Theta}\, \frac{1}{n} \sum_{t=1}^n l_{\alpha}(X_t;\T):=\argmin_{\theta \in \Theta} \frac{1}{n} H_{\alpha,n}(\theta),
\end{eqnarray}
where \begin{eqnarray*}
l_\alpha(X_t;\theta) = \left\{ \begin{array}{ll}
   \displaystyle  \int f_\theta^{1+\alpha}(x) dx - \left( 1 + \frac{1}{\alpha} \right)
     f_\theta^{\alpha}(X_t)    & \mbox{, $\alpha > 0$,}\vspace{0.15cm}\\
   \displaystyle  - \log f_\theta(X_t)      & \mbox{, $\alpha = 0$.}
   \end{array}
 \right.
\end{eqnarray*}
As was well demonstrated in \cite{BHHJ:1998},  the efficiency of the estimator gets closer to that of the MLE as $\A$ decreases to 0 and it has strong robustness when $\A$ increases. That is, the tuning parameter $\alpha$ controls the trade-off between efficiency and robustness. In particular, the MDPDE with $\A$ close to 0 was found to tend to enjoy both high efficiency and robustness. Another virtue of the MDPD estimation procedure is that it can be conventionally applied to other parametric models, resulting in  robust and efficient estimators in various models. See, for example, \cite{ghosh.et.al:2013}, \cite{dierckx.et.al:2013}, and   \cite{song:2017}.

Throughout the paper, a major focus is made on  the change point problem in the parametric framework. Hence, we assume that $g$ belongs to   $\{ f_{\theta} \}$, that is, $g=f_{\T_0}$ for some $\T_0 \in \Theta$. To establish the limiting behaviors of the stopping time below,
we introduce some assumptions. Particularly, the following three assumptions are made to ensure the strong consistency of $\hat{\T}_{\A,n}$.  We assume that $\theta_0$  belongs to the parameter space $\Theta$.
\begin{enumerate}
\item[\bf A1.] The parameter space $\Theta$ is a compact subset in $\mathbb{R}^d$.
\item[\bf A2.] The density $f_\theta$ and  the integral $\int f_\T^{1+\A}(x)dx$ are continuous in $\theta$.
\item[\bf A3.] There exists a function $B(x)$ such that $|l_\A(x;\T)| \leq B(x)$ for all $x$ and $\T$ and $\E[B(X_t)]<\infty$.
\end{enumerate}
Assumption {\bf A3} can be replaced with $\E \sup_{\T\in\Theta} |l_\A (X_t) | <\infty$.
By assumption {\bf A2}, $l_\A(x;\T)$ becomes a continuous function in $\T$, and thus it follows from assumption {\bf A3} that
\[\sup_{\T \in \Theta} \bigg|\frac{1}{n}\sum_{t=1}^n l_\A(X_t;\T) -\E[l_\A(X_t;\T)] \bigg| = o(1)\quad a.s.\]
(cf. chapter 16 in \cite{ferguson:1996} or Theorem 2.7 of \cite{straumann:mikosch:2006}). Noting the fact that $\E[l_\A (X_t;\T)] =d_\A (f_{\T_0},f_\T)-\frac{1}{\A}\int f_{\T_0}(x)dx$, one can see that  $\E[l_\A (X_t;\T)]$ has a unique minimum at $\T_0$. Thus, $\hat\theta_{\A,n}$ converges almost surely to $\theta_0$ by the standard arguments.
\begin{thm}\label{thm1}
Suppose that assumptions {\bf A1}-{\bf A3} hold. Then, for each $\A\geq0$, $\hat\theta_{\A,n}$ converges to $\theta_0$ almost surely.
\end{thm}
\vspace{0cm}
\begin{remark}
Assumption {\bf A3} is ensured by such condition that $\sup_{x,\T \in \Theta} f_\T(x) <\infty$, which is usually obtained by restricting the range of scale parameter. For example, when the normal parametric family $\{N(\mu,\sigma^2)\}$ is considered, the condition is satisfied by considering the parameter space $\Theta=\{(\mu,\sigma)|\ -\infty < \mu<\infty, \sigma\geq c\}$ for some $c>0$.
\end{remark}

In the following subsection, we will introduce a detector and a stopping time for monitoring change in parameters using the MDPDE above and its objective function, and then investigate its asymptotic behaviors. Hereafter, $\| \cdot\|$ denotes any  convenient vector norm  and will also be used to denote induced matrix  norm without causing confusion.  That is, for a square matrix $A$ and a vector $x$, $\|A\|=\sup_{\|x\|=1} \|Ax\|$.  For notational convenience, we use $\pa_{\theta}$ and $\paa$ to denote $\frac{\pa}{\pa\T}$ and $\frac{\pa^2}{\pa\theta \pa \theta'}$, respectively.

\subsection{DPD based sequential change point test under the null hypothesis}

Let $\{X_1, \cdots, X_n\}$ be {\it i.i.d.} observations from a density $f_{\theta_0}$ and suppose that we sequentially observe $X_{n+1}, X_{n+2},$ $\cdots.$ Here, the historical data $\{X_1, \cdots, X_n\}$ is assumed not to undergo parameter change. At a monitoring instant $t=n+k$, we wish to test the following null hypothesis, particularly in the presence of outliers:
\[H_0: \theta_0\text{ does not change over } t \leq n+k\]
against
\[ H_1: \theta_0\text{ changes at some time }  t \leq n+k.\]

Sequential procedure for testing the hypotheses above consists of a detecting statistics, which is  called detector, and a boundary function. When a detector crosses a pre-given  boundary function, the sequential procedure gives an alarm signal indicating a parameter change. In the literature, mainly two types of detectors have been proposed: fluctuation based detectors and CUSUM based detectors. The first type detectors compare the estimates obtained in historical data and arrived data to calculate the fluctuations of the estimates. This approach was introduced by  \cite{chu.et.al:1996} and has been generalized, for example,  by \cite{leisch.et.al:2000} and \cite{bardet:kengne:2014}. The second one considered by  \cite{chu.et.al:1996}, \cite{horvath.et.al:2004} and \cite{aue.et.al:2006} usually use the CUSUM values of residuals to measure the stability of the fitted models.
The score function based detector was also proposed by \cite{berkes.et.al:2004seq}, which indeed can be considered as a CUSUM type detector. 
 For more details on the sequential procedure, we refer the readers to the seminal paper of \cite{chu.et.al:1996}.

In this study, we use the objective function of the MDPDE to propose a CUSUM type detector as follows: for each $\A\geq0$,
\begin{eqnarray}\label{detector}
D_{\A,n}(k):=\frac{ \big\| \hat{\mathcal{I}}_{\A,n}^{-\frac{1}{2}}\, \pa_\theta H_{\A,n+k}(\hat\theta_{\A,n})\big\|}{\sqrt{n} \Big(1+\frac{k}{n}\Big)},
\end{eqnarray}
where  $\hat\theta_{\A,n}$ and $\hat{\mathcal{I}}_{\A,n}$ are the MDPDE in (\ref{MDPDE}) and a consistent estimator for $\mathcal{I}_\A$ defined in assumption {\bf A6} below, respectively, and both are obtained from the historical data $\{X_1, \cdots, X_n\}$.
Now, let $b(\cdot)$ be a boundary function. Then, the sequential procedure is stopped as soon as $D_{\A,n} (k)$ crosses $b(k/n)$, rejecting $H_0$. In other words, the stopping time of the procedure is defined as the first hitting time of the detector and  may then be
expressed as
\begin{eqnarray}\label{stop}
k_{\A,n}:=\min \Big\{ k\geq1 : \big\| \hat{\mathcal{I}}_{\A,n}^{-\frac{1}{2}}\, \pa_\theta H_{\A,n+k}(\hat\theta_{\A,n})\big\| > \sqrt{n} \Big(1+\frac{k}{n}\Big) b\Big(\frac{k}{n}\Big)\Big\}.
\end{eqnarray}
 If $k_{\A,n}<\infty$,  $H_0$ is rejected and we conclude that $\theta_0$ changed at some time $t\leq n+k_{\A,n}$. 
\begin{remark}\label{s.detector}
Since $\pa_{\theta}H_{\A,n}(\hat{\theta}_{\A,n})=0$, the proposed detector can be expressed as
\[D_{\A,n}(k)=\frac{ \Big\| \hat{\mathcal{I}}_{\A,n}^{-\frac{1}{2}}\, {\displaystyle \sum_{t=n+1}^{n+k}} \pa_\theta  l_\A(X_t;\hat\theta_{\A,n})\Big\|}{\sqrt{n} \Big(1+\frac{k}{n}\Big)}.\]
Here, we note that $D_{\A,n}$ with $\A=0$ becomes the score based detector. This type of detector was introduced by  \cite{berkes.et.al:2004seq}. The score function is deduced from the KL divergence, so the detector $D_{\A,n}$ with $\A>0$ can be thought of as a DP divergence version of the score based detector. Thus, the proposed procedure is expected to be able to adjust robustness and efficiency  by controlling the tuning parameter $\A$ as does the MDPDE.
\end{remark}
\begin{remark}
The robustness of the detector with $\A>0$ is obtained via the following two steps. First, as is well recognized in the previous studies, the MDPDE $\hat\theta_{\A,n}$ with $\A>0$ estimates the true parameter robustly in the presence of outliers. Obviously, this must be a  basic requirement in constructing a robust statistics. It is, however, not sufficient to make the sequential procedure robust against outliers. Note that the existing detectors are calculated from the observations, $X_1,\cdots, X_{n+k}$, and the estimates, $\hat\theta_n$. This means that the value of the detector can be distorted by outlying observations even if  the true parameter is properly estimated. Thus, additional measures are needed to lessen the impact of outliers on the detectors. Our second step for the robustness is from  the term $\pa_\theta H_{\A,k}$  in the proposed detector. This term gives a down-weight to the outlying observations, resulting in reducing the influence of outliers on our detector. This is actually similar to how the MDPDE has the robust property (cf. \cite{BHHJ:1998}). To see this, observe that
\begin{eqnarray*}
 \pa_\T H_{\A,k}(\theta) = \left\{ \begin{array}{ll}
   \displaystyle k(1+\A)\int U_\theta(z) f_\theta^{1+\A}(z)dz-\sum_{t=1}^k f_\theta^\A(X_t)U_\theta(X_t)   & \mbox{, $\alpha > 0$,}\vspace{0.15cm}\\
   \displaystyle  - \sum_{t=1}^k  U_\theta(X_t)      & \mbox{, $\alpha = 0$.}
   \end{array}
 \right.
\end{eqnarray*}where $U_\theta(x)=\pa_\T \log f_\theta(x)$. By comparing with  $\pa_\theta H_{\A,k}(\T)$ with $\A=0$ in the score based detector, one can see that $\pa_\theta H_{\A,k}(\T)$ with $\A>0$ provides density power weight, $f_\theta^\A (X_t)$, to each $U_\theta(X_t)$, whereas  $\pa_\theta H_{0,k}(\T)$ gives equal weight.
In sum, the robustness of the sequential procedure based on $D_{\A,n}(k)$ with $\A>0$ is achieved via the robust estimator $\hat\theta_{\A,n}$ and the term $\pa_\theta H_{\A,k}(\T)$ that gives a  down-weight to outliers.
\end{remark}


In order to perform the sequential procedure with the stopping time $k_{\A,n}$, it needs to derive the following limiting probability:
\begin{eqnarray*}
\lim_{n\rightarrow\infty} P\big( k_{\A,n} <\infty\ |\ H_0 \big),
\end{eqnarray*}
which indeed represents the type I error of the test procedure. This limiting value is expressed as a function of the boundary function, so the boundary function may be chosen such that the above value is equal to a given significance level. For example, as in Remark \ref{rm.max} below, when a constant boundary function, i.e., $b(\cdot)=b$, is employed, one can calculate the critical value $b$ by solving the equation. 

 To establish the limiting behavior of $k_{\A,n}$, particularly in the case of $\A>0$, we impose further assumptions.
\begin{enumerate}
\item[\bf A4.] The integral $\int f_\T^{1+\A}(z)dz$ is differentiable two times with respect to $\T$ and the derivative can be taken under the integral sign.
\item[\bf A5.] The true parameter $\T_0$ is in the interior of $\Theta$.
\item[\bf A6.] $\paa l_\A (x;\T)$ is continuous in $\T$ and there exists an open neighborhood $N(\T_0)$ of $\T_0$ such that $\E\big[ \sup_{\T \in N(\T_0)} \big\| \paa l_\A (X_t;\T)\big\|\big] <\infty$.
\item[\bf A7.] The matrices $\mathcal{J}_\A$ and $\mathcal{I}_\A$ defined by
\begin{eqnarray*}
\mathcal{I}_\A&:=&\E \big[ \pa_\theta l_\A(X_t;\T_0) \pa_{\theta'} l_\A(X_t;\T_0)\big],\\
 \mathcal{J}_\A&:=&\E \big[ \paa l_\A(X_t;\T_0)\big]=(1+\A)\int f_{\T_0}^{\A-1}(z) \pa_\T f_{\T_0}(z)  \pa_{\T'} f_{\T_0}(z) \ud z
\end{eqnarray*}
exist and are positive definite.
\end{enumerate}
We also assume the boundary function satisfies the following condition:
\begin{enumerate}
\item[{\bf B}.] The boundary function $b(\cdot)$ is continuous on $(0,\infty)$ and $\inf_{t>0} b(t) >0$.
\end{enumerate}
The following  theorem is the first result in this study.
\begin{thm}\label{thm2}
Suppose that assumptions {\bf A1}--{\bf A7} and {\bf B} hold. Then, we
have that for each $\A\geq0$,
\begin{eqnarray*}
\lim_{n\rightarrow\infty} P\big( k_{\A,n} <\infty\ |\ H_0 \big)=P\left(  \sup_{0<s<1} \frac{ \big\| W_d(s)\big\|}{b(s/(1-s))} >1 \right),
\end{eqnarray*}
where $\{W_d(s)\}$ is a  $d$-dimensional standard Wiener process.
\end{thm}
\begin{remark}\label{rm.max}
If the maximum norm and a constant boundary functions, $b(\cdot)=b$, are considered,  then we have
\begin{eqnarray*}\label{CB}
\lim_{n\rightarrow\infty} P\big( k_{\A,n} <\infty\ |\ H_0 \big)&=&1 - \Big[ P\Big( \sup_{0<s<1} |W(s)| \leq b \Big) \Big]^d \nonumber\\
&=&1-\left\{ \frac{4}{\pi}\sum_{k=0}^\infty\frac{(-1)^k}{2k+1}\exp\bigg(-\frac{\pi^2(2k+1)^2}{8b^2}\bigg)\right\}^d.
\end{eqnarray*}
Since the summands in the infinite series decrease exponentially  with increasing $k$, the critical value $b$ can be accurately computed by numerical method.
The critical values corresponding to the significance levels of 1\%, 5\%, and 10\% are given in Table \ref{tab:cr.value}.
\end{remark}
\begin{table}[H]
  \centering
  {\small
  \caption{Critical values at the significance levels of $\A$=1\%, 5\% and 10\%.}\vspace{0.1cm}
  \tabcolsep=7.5pt
  \renewcommand{\arraystretch}{1} \label{tab:cr.value}
    \begin{tabular}{cllllllllll}
    \toprule
       & \multicolumn{10}{l}{$d$} \\
\cmidrule{2-11}        $\A$     & 1     & 2     & 3     & 4     & 5     & 6     & 7     & 8     & 9     & 10 \\
    \midrule
    1\%   & 2.807 & 3.023 & 3.143 & 3.226 & 3.289 & 3.340 & 3.383 & 3.419 & 3.451 & 3.480 \\
    5\%   & 2.241 & 2.493 & 2.632 & 2.728 & 2.800 & 2.859 & 2.907 & 2.948 & 2.984 & 3.016 \\
    10\%  & 1.960 & 2.231 & 2.381 & 2.484 & 2.561 & 2.623 & 2.675 & 2.719 & 2.758 & 2.792 \\
    \bottomrule
    \end{tabular}} 
\end{table}%

\subsection{DPD based sequential change test under the alternative hypothesis}

Let $\{X_{0,t} | t\in\mathbb{N}\}$  and  $\{X_{1,t} | t\in\mathbb{N}\}$ be the sequences of {\it i.i.d.} random variables from $f_{\T_0}$ and $f_{\T_1}$, respectively, where $\T_0\neq \T_1$.  Denote the observations up to the monitoring time by $\{X_1, \cdots,X_{n+k}\}$ and consider the following alternative hypothesis: for some fixed $k^*>0$,
\begin{eqnarray*}
H_1: X_t=\left\{\begin{split}
 X_{0,t}, \quad &t=1,\cdots, n+k^*.\\
 X_{1,t},\quad  &t=n+k^*+1,\cdots ,n+k.
 \end{split}\right.
\end{eqnarray*}
The alternative hypothesis represents a situation  that the historical data $\{X_1,\cdots,X_n\}$ follows from the density $f_{\T_0}$ and the parameter changes from $\T_0$ to $\T_1$ at $t=n+k^*+1$. To establish the limiting behavior of $k_{\A,n}$ under $H_1$, we assume the followings:
\begin{enumerate}
\item[\bf A8.] For some closed  neighborhood $N'(\T_0)$ of $\T_0$,   $\E\big[ \sup_{\T \in N'(\T_0)} \big\| \pa_\T l_\A (X_{1,t};\T)\big\|\big] <\infty$.
\item[\bf A9.] $\big\|  E \big[\pa_{\theta}\,l_\A(X_{1,t};\T_0)\big]\big\|> 0$.
\end{enumerate}
\begin{thm}\label{thm3}
Suppose that assumptions {\bf A1}--{\bf A9} and {\bf B} hold.  If $\sup_{t>0} b(t) <\infty$, then we have that for each $\A\geq0$,
\begin{eqnarray*}
\lim_{n\rightarrow\infty} P\big( k_{\A,n} <\infty\ |\ H_1 \big)=1.
\end{eqnarray*}
\end{thm}



\begin{remark}\label{choice.A} Selection of an optimal $\A$ can be an issue in actual practice.
In the sequential test procedure, parameters are assumed not to change in the historical data, so one could consider to use the existing selection criteria such as \cite{warwick:2005}, \cite{fugisawa:eguchi:2006}, and \cite{durio:isaia:2011} to determine an optimal $\A$. However, it could not be appropriate because the degree of contamination in a monitoring period can be different from that in the historical data.  As an illustrative example,  consider a situation  that outliers are not included in the historical data but the arriving  data is contaminated by outliers. In this case, $\A=0$ is usually chosen as the optimal $\A$, but the sequential procedure implemented with $\A=0$ is likely to be distorted by outliers occurring in the monitoring period. In short, determining $\A$ in advance may or may not be helpful depending on the situation, so it seems difficult to establish a systematic selection rule.   Nevertheless,  we recommend practitioners to use an $\A$ in $[0.2,0.3]$ based on our simulation results.  This is because too large $\A$ such as $\A\geq 0.5$ can  lead to a significance loss in powers and the procedure with $\A \leq 0.1$ can be influenced to some degree by outliers.  The procedure with  an $\A$  in [0.2,0.3] shows strong enough robustness while performing as efficiently as the score based procedure (see Section \ref{Sec:sim}).
\end{remark}

As aforementioned in Section \ref{Sec:intro}, the MDPDE can be conveniently applied to various parametric models including time series models and multivariate models.   
Once such MDPDE is set up, our sequential test procedure can be extended to the corresponding models. In the following section, we develop a sequential change point test for time series models based on the DP divergence. As an application, the sequential test procedure for GARCH models are provided. All the remarks mentioned in this section still hold for the extended cases.

\section{DP divergence based sequential change point test in time series models}\label{Sec:3}
We first briefly introduce MDPD estimation procedure in time series models. However, since we are focused not on estimation but on sequential testing procedure, the asymptotic properties of the MDPDE for time series models are not dealt with in detail. Indeed, the MDPD estimation procedure has been already applied to various time series models, for example, in GARCH models, multivariate Gaussian time series models, and time series models of counts, and  its strong consistency and asymptotic normality were well established. See, for example, \cite{lee:song:2009}, \cite{kim:lee:2011}, and \cite{kang:lee:2014}. Hence, in this section, we assume that the strong consistency and $\sqrt{n}$-consistency are established, and aim to provide some sufficient conditions for deriving the asymptotic behaviors of our stopping time below. For simplicity of presentation, we consider univariate time series model.

\subsection{MDPDE for time series models}\label{Sub3:1}
 Assume that a time series $\{X_t | t\in\mathbb{Z}\}$ is strictly stationary and ergodic, and that  the conditional mean and variance are specified by a $d$-dimensional parameter vector $\T\in\Theta$. Then, the time series model that generates $\{X_t\}$ may be represented by
 \begin{eqnarray}\label{TS}
X_t=\mu_t(\theta)+\sigma_t(\theta)\epsilon_t,
\end{eqnarray}
where  $\mu_t(\T)=E(X_t|\mathcal{F}_{t-1})$, $\sigma_t^2(\T)=Var(X_t|\mathcal{F}_{t-1})$,  where $\mathcal{F}_t=\sigma(X_s|s\leq t)$, and $\{\epsilon_t|t\in\mathbb{Z}\}$ is a sequence of i.i.d. random variables from a density $f_\ep$ with zero mean and unit variance. Various times series models such as ARMA models, GARCH-type models, and ARMA-GARCH models are included in the model above. It is noteworthy that, due to Theorem 20.1 in \cite{billingsley:1995}, each $\mu_t(\T)$ and $\sigma_t(\T)$ can be expressed as a measurable function of $\{X_{t-1},X_{t-2},\cdots\}$ and $\T$. We assume that the parameter space $\Theta$ is a compact subset of $\mathbb{R}^d$ and the true parameter $\T_0$ is in the interior of $\Theta$. To apply MDPD estimation procedure, we consider the situation that the error distribution $f_\ep$ is specified. Hereafter, we denote the process from the above model with $\T$ by $\{X_t^{\T}|t\in\mathbb{Z}\}$ and $X_t^{\T_0}$ is simply denoted by $X_t$ except when $\T_0$ changes in the observations.


Let $f_\T (x|\mathcal{F}_{t-1})$ be the conditional density of $X_t^\T$ given $\mathcal{F}_{t-1}$.
Then, the DP divergence between the two conditional densities $f_{\T_0} (x|\mathcal{F}_{t-1})$ and $f_\T (x|\mathcal{F}_{t-1})$ is defined by
\begin{eqnarray*}
&&d_\alpha\left(f_{\theta_0}(\cdot|\mathcal{F}_{t-1}),f_\theta(\cdot|\mathcal{F}_{t-1})\right)\\
&&\\
&&=\left\{\begin{array}{ll}
                         \displaystyle \int\left\{ f_\theta^{1+\A}(x|\mathcal{F}_{t-1})-(1+\frac{1}{\A})f_{\theta_0}(x|\mathcal{F}_{t-1})f_\theta^\A(x|\mathcal{F}_{t-1})+\frac{1}{\A}f_{\theta_0}^{1+\alpha}(x|\mathcal{F}_{t-1}) \right\}dx &, \A>0 \vspace{0.3cm}\\
                          \displaystyle \int f_{\theta_0}(x|\mathcal{F}_{t-1})
                           \left\{\log f_{\theta_0}(x|\mathcal{F}_{t-1})-\log
                           f_\theta(x|\mathcal{F}_{t-1})\right\}dx&,\A=0.
                        \end{array}\right.
\end{eqnarray*}
Given the observations $X_1,\cdots, X_n$, the MDPDE for the model (\ref{TS}) is obtained in the same manner as in (\ref{MDPDE}):
\begin{eqnarray*}
\hat \theta_{\alpha, n}^o = \argmin_{\theta \in \Theta}\, \frac{1}{n} \sum_{t=1}^n l_{\alpha}(X_t;\T):=\argmin_{\theta \in \Theta} \frac{1}{n} H_{\alpha,n}(\theta),
\end{eqnarray*}
where
\begin{eqnarray*}
l_\A(X_t;\theta)=\left\{\begin{array}{ll}
                          \displaystyle\int f_{\theta}^{1+\alpha}(x|\mathcal{F}_{t-1})dx-\left(1+\frac{1}{\alpha}\right)f_{\theta}^\alpha(X_t|\mathcal{F}_{t-1}) &, \alpha>0\vspace{0.3cm} \\
                           -\log f_{\theta}(X_t|\mathcal{F}_{t-1}).
                           &,\alpha=0.
                        \end{array}\right.
\end{eqnarray*}
We note that due to (\ref{TS}), the conditional density is given by
\[f_{\theta}(x|\mathcal{F}_{t-1}) =\frac{1}{\sigma_t(\T)}f_\ep \bigg( \frac{x-\mu_t(\T)}{\sigma_t(\T)} \bigg).  \]
Without confusion, we share the same notations $l_\A(X_t;\T)$, $H_{\A,n}(\T), \mathcal{I}_\A$, and $\mathcal{J}_\A$ in Section \ref{Sec:2}. Throughout this section, $l_\A(x;\T)$ is assumed to be twice continuously differentiable with respect to $\T$.
We impose the further assumptions on $\{l_\A(X_t;\T)|t\in\mathbb{Z}\}$ as follows: for each $\A\geq0$,
\begin{enumerate}
\item[\bf M1.] $\{l_\A(X_t;\T)|t\in\mathbb{Z}\}$ is strictly stationary and ergodic for each $\T\in\Theta$.
\item[\bf M2.]  $\mathcal{I}_\A=\E \big[ \pa_\theta l_\A(X_t;\T_0) \pa_{\theta'} l_\A(X_t;\T_0)\big]$ and $\mathcal{J}_\A=\E \big[ \paa l_\A(X_t;\T_0)\big]$ are non-singular.
\end{enumerate}
\begin{remark}
Assumption {\bf M1} is ensured by the strictly stationarity and ergodicity of $\{\mu_t(\T)|t\in\mathbb{Z}\}$ and $\{\sigma_t^2(\T)|t\in\mathbb{Z}\}$. For example, in GARCH models, assumption {\bf G1} in subsection \ref{subsec:GARCH} below implies the stationarity and ergodicity of $\{\sigma_t^2(\T)|t\in\mathbb{Z}\}$.
\end{remark}
When estimating the model (\ref{TS}), $\mu_t(\T)$ and $\sigma_t(\T)$ are usually not explicitly obtained because of the initial value issue.  In this case, they are generally approximated by using a recursion with some initial values chosen. In GARCH models and  ARMA-GARCH models, one can find the approximated processes, for exmple, in \cite{berkes.et.al:2003} and \cite{francq:zakoian:2004}. Hereafter, we denote the approximated processes by $\{ \tilde \mu_t(\T) |t=1\cdots,n\}$ and $\{ \tilde \sigma^2_t(\T) |t=1\cdots,n\}$. Then the MDPDE actually used is given by
\begin{eqnarray}\label{MDPDE.ts}
\hat \theta_{\alpha, n} = \argmin_{\theta \in \Theta}\, \frac{1}{n} \sum_{t=1}^n \tilde l_{\alpha}(X_t;\T):=\argmin_{\theta \in \Theta} \frac{1}{n} \tilde H_{\alpha,n}(\theta),
\end{eqnarray}
where
\begin{eqnarray*}
\tilde l_\A(X_t;\theta)=\left\{\begin{array}{ll}
                          \displaystyle \frac{1}{\tilde \sigma_t^{\A+1}(\T)} \int f_\ep^{1+\A} \bigg( \frac{x-\tilde \mu_t(\T)}{\tilde \sigma_t(\T)} \bigg)dx-\left(1+\frac{1}{\alpha}\right)\frac{1}{\tilde\sigma^\A_t(\T)}f^\A_\ep \bigg( \frac{X_t-\tilde \mu_t(\T)}{\tilde \sigma_t(\T)}\bigg) &, \alpha>0\vspace{0.3cm} \\
                          \displaystyle
                           -\log \frac{1}{\tilde\sigma_t(\T)}f_\ep \bigg( \frac{X_t-\tilde \mu_t(\T)}{\tilde \sigma_t(\T)}\bigg)
                           &,\alpha=0.
                        \end{array}\right.
\end{eqnarray*}
\begin{remark} In the case that the support of $f_\ep$  is $\mathbb{R}$, we can see that
\[\int f_\ep^{1+\A} \bigg( \frac{x-\tilde \mu_t(\T)}{\tilde \sigma_t(\T)} \bigg)dx=\tilde\sigma_t(\T) \int f_\ep^{1+\A}(x)dx.\]
Since $\int f^{1+\A}_\ep(x)dx$ does not include any parameters, for each $\A$, one can calculate the integral in advance  analytically or numerically. For example, if $\ep_t$ follows $N(0,1)$, the integral is equal to $(1+\A)^{-1/2}$. Hence, through replacing $\int f_\ep^{1+\A}(x)dx$ with the pre-obtained value, the computational burden in optimizing the above  objective function can be significantly reduced.
\end{remark}
\noindent As mentioned above, the strong consistency and $\sqrt{n}$--consistency of $\hat\T_{\A,n}$ is assumed to hold:
\begin{enumerate}
\item[\bf M3.] $\hat\T_{\A,n}$ converges almost surely to $\T_0$ and $\sqrt{n}(\hat\T_{\A,n}-\T_0)=O_P(1)$.
\end{enumerate}
Since $d_\alpha(f_{\theta_0}(\cdot|\mathcal{F}_{t-1}),f_\theta(\cdot|\mathcal{F}_{t-1}))$  is almost surely nonnegative and minimized at $\T_0$, the minimum of $\E[ d_\alpha(f_{\theta_0}(\cdot|\mathcal{F}_{t-1}),f_\theta(\cdot|\mathcal{F}_{t-1}))]$ is obtained also at $\T_0$. Noting that
\begin{eqnarray}\label{EL}
\E[l_\A (X_t;\T)] =\E[ d_\alpha(f_{\theta_0}(\cdot|\mathcal{F}_{t-1}),f_\theta(\cdot|\mathcal{F}_{t-1}))]-\frac{1}{\A}\E \Big[\int f_{\T_0}(x|\mathcal{F}_{t-1})dx\Big],
\end{eqnarray}
one can see that $\E[l_\A (X_t;\T)]$ has the minimum value at the true parameter $\T_0$. Therefore, since $\{l_\A(X_t;\T)|t\geq1\}$ is strictly stationary and ergodic by assumption {\bf M1}, the strong consistency of $\hat\T_{\A,n}$ can be shown by the standard arguments, for example, if the followings are satisfied:
\begin{eqnarray*}
\E \sup_{\T\in\Theta} l_\A(X_t;\T)<\infty\quad\mbox{and}\quad\frac{1}{n}\sum_{t=1}^n \sup_{\T\in\Theta} \big| l_\A(X_t;\theta)-\tilde l_\A(X_t;\theta)\big|=o(1)\quad a.s.
\end{eqnarray*}
When fitting a model for which the MDPDE has not yet been established, one can verify the strong consistency by checking that
 the two conditions above hold for the model under consideration. Although we assume $\sqrt{n}$--consistency of $\hat\T_{\A,n}$ in assumption {\bf M3}, this can indeed be derived under  the conditions {\bf S1}--{\bf S4} introduced in the following subsection. For more details, see Lemma \ref{root.n.con} in Appendix.
\subsection{DPD based sequential test in time series models}\label{Sub3:2}
Suppose that $\{X_1, \cdots, X_n\}$ is a historical data from (\ref{TS}) with $\theta_0$ and $X_{n+1}, X_{n+2},$ $\cdots$ are being observed sequentially. At monitoring time $n+k$, we wish to test the following null hypothesis:
\begin{eqnarray*}
 H_0: \theta_0\text{ does not change over } t \leq n+k.
\end{eqnarray*}
Let $\hat{\mathcal{I}}_{\A,n}$ be a consistent estimator for $\mathcal{I}_\A $ obtained from  the historical data. By assumption {\bf M2}, $\hat{\mathcal{I}}_{\A,n}$ is invertible for sufficiently large $n$. Hence, similarly to (\ref{detector}) and (\ref{stop}), we can define a DPD based  detector and stopping time as follows:
\begin{eqnarray}\label{detector.ts}
\tilde D_{\A,n}(k):=\frac{ \big\| \hat{\mathcal{I}}_{\A,n}^{-\frac{1}{2}}\, \pa_\theta \tilde H_{\A,n+k}(\hat\theta_{\A,n})\big\|}{\sqrt{n} \Big(1+\frac{k}{n}\Big)}
\end{eqnarray}
and
\begin{eqnarray}\label{stop.ts}
\tilde{k}_{\A,n}:=\min \Big\{ k\geq1 \big|\, \tilde D_{\A,n}(k) >  b\Big(\frac{k}{n}\Big)\Big\},
\end{eqnarray}
where $\hat\T_{\A,n}$ is the MDPDE defined in (\ref{MDPDE.ts}) and obtained also from the historical data.  The boundary function $b(\cdot)$ is assumed to satisfy assumption ${\bf B}$ in Section \ref{Sec:2}.
We now present the following sufficient conditions under which the result in Theorem \ref{thm2} holds also for the stopping time $\tilde k_{\A,n}$.
\begin{enumerate}
\item[\bf S1.] $\{\pa_\T l_\A(X_t ;\T_0)\}$ is a martingale difference with respect to $\mathcal{F}_t$.
\item[\bf S2.]$\displaystyle \frac{1}{\sqrt{n}}\sum_{t=1}^n \big\| \pa_{\theta}\,l_\A(X_t;\theta_0)-\pa_{\theta}\,\tilde l_\A(X_t;\theta_0)\big\|= o(1)\quad a.s.$
\item[\bf S3.]For some neighborhood $N_1(\T_0)$ of $\T_0$, $$\displaystyle \frac{1}{n}\sum_{t=1}^{n}\sup_{\theta \in N_1(\T_0)}\big\|\paa  l_\A(X_t;\theta)
-\paa  \tilde l_\A(X_t;\theta)\big\|= o(1)\quad a.s.$$
\item[\bf S4.] For some neighborhood $N_2(\T_0)$ of $\T_0$,
 $$\displaystyle\E \sup_{\theta \in N_2(\T_0)}\big\|\paa  l_\A(X_t;\T)\big\|  <\infty.$$
\end{enumerate}
Condition {\bf S1} is required to use the functional central limit
 theorem (FCLT) and apply the Hájek-Rényi-Chow inequality in Lemma \ref{lm.I3} below.
This condition  is commonly satisfied in stationary time series models. In particular, if the derivative can be taken under the integral in $l_\A(X_t;\T)$, the condition is deduced from the fact that  \[\E\big[f_{\T_0}^{\A-1}(X_t|\mathcal{F}_{t-1})\pa_{\T}f_{\T_0}(X_t|\mathcal{F}_{t-1})|\mathcal{F}_{t-1}\big]
=\int f_{\T_0}^\A(x|\mathcal{F}_{t-1})\,\pa_\T f_{\T_0}(x|\mathcal{F}_{t-1})dx.\]
 Condition {\bf S4} together with  the continuity of $\paa l_\A(x;\T)$ is used to verify Lemma \ref{lm.I0}, which are very helpful
for proving Lemma \ref{lm.I1}.
As will be seen in Appendix, Lemmas \ref{lm.I1}-\ref{lm.I3} are key lemmas for deriving the asymptotic behavior of $\tilde k_{\A,n}$.
In most of  time series models including GARCH-type terms, similar conditions like {\bf S2} and  {\bf S3}, usually the case of $\A=0$, are often established and usefully used in  deriving asymptotic properties of the estimator such as QMLE.
\begin{thm}\label{thm4}
Suppose that assumptions  {\bf M1}--{\bf M3}, {\bf B}, and conditions  {\bf S1}-{\bf S4} hold. Then, we
have that for each $\A\geq0$,
\begin{eqnarray*}
\lim_{n\rightarrow\infty} P\big( \tilde k_{\A,n} <\infty\ |\ H_0 \big)=P\left(  \sup_{0<s<1} \frac{ \big\| W_d(s)\big\|}{b(s/(1-s))} >1 \right),
\end{eqnarray*}
where $\{W_d(s)\}$ is a  $d$-dimensional standard Wiener process.
\end{thm}
Next, we establish the asymptotic property of $\tilde k_{\A,n}$ under the alternative hypothesis below. To be more specific, let $\{X^{\T_0}_t | t\in\mathbb{Z}\}$  and  $\{X^{\T_1}_t | t\in\mathbb{Z}\}$ be the strictly stationary and ergodic processes from the model (\ref{TS}) with the parameter $\theta_0$ and $\theta_1(\neq \theta_0)$, respectively.  Then, the alternative hypothesis under consideration is expressed as follows  : for some fixed $k^*>0$,
\begin{eqnarray*}
H_1: X_t=\left\{\begin{split}
 X^{\T_0}_t, \quad &t=1,\cdots, n+k^*\\
 X^{\T_1}_t,\quad  &t=n+k^*+1,\cdots ,n+k.
 \end{split}\right.
\end{eqnarray*}
To establish the consistency of the test procedure, we assume that assumptions {\bf M1}  hold for
 $\{l_\A(X_t^{\T_1};\T)|t\in\mathbb{Z}\}$. Additionally, the following conditions are made: for some closed neighborhood $N_3(\T_0)$ of $\T_0$,
\begin{enumerate}
\item[\bf S5.]  $\displaystyle \E \sup_{\T \in N_3(\T_0)} \big\| \pa_\T\, l_\A (X^{\T_1}_t;\T)\big\| <\infty$.
\item[\bf S6.]$\displaystyle \frac{1}{k_n-k^*} \sum_{t>n+k^*}^{n+k_n}\sup_{\theta \in N_3(\T_0)}\big\| \pa_{\theta}\,l_\A(X^{\T_1}_t;\T)-\pa_{\theta}\,\tilde l_\A(X^{\T_1}_t;\T)\big\|= o(1)\ a.s.,\vspace{0.2cm}$ where $\{k_n\}$ is an increasing sequnce of positive integers.
\end{enumerate}
Like assumption {\bf A9}, we impose the following key assumption for establishing the consistency of the test procedure. Indeed, $\E\,  l_\A(X^{\T_1}_t;\T)$  has the minimum value at $\T_1$, so if the minimizer is unique,  the following condition is obviously fulfilled provided that the differentiation and the expectation are interchangeable.
\begin{enumerate}
\item[\bf S7.] $E \big\| \pa_{\theta}\,l_\A(X^{\T_1}_t;\T_0)\big\|>0.$
\end{enumerate}
\begin{thm}\label{thm5}
Suppose that assumptions in Theorem \ref{thm4} still hold for $\{X^{\T_0}_t | t\geq1\}$ and that assumptions {\bf M1} hold for $\{l_\A(X_t^{\T_1};\T)|t\in\mathbb{Z}\}$. If conditions {\bf S5}--{\bf S7} are satisfied and $\sup_{t>0} b(t) <\infty$, then we have that for each $\A\geq0$,
\begin{eqnarray*}
\lim_{n\rightarrow\infty} P\big( \tilde k_{\A,n} <\infty\ |\ H_1 \big)=1.
\end{eqnarray*}
\end{thm}
\begin{remark} Let $\{k_n\}$ be an increasing sequence of positive integers satisfying $k_n/\sqrt{n}\rightarrow\infty$. In the proof of Theorem \ref{thm5}, we show that
\begin{eqnarray*}
\frac{\big\| \pa_\T \tilde H_{\A,n+k_n}(\hat{\theta}_{\A,n})\big\|}{\sqrt{n}\Big(1+\frac{k_n}{n}\Big)}
&\stackrel{P}{\longrightarrow}& \infty,
\end{eqnarray*}
which implies that $P(\tilde k_{\A,n} -k^* \leq k_n-k^*)\rightarrow 1$. Hence, we can see that the detection delay $\tilde k_{\A,n} -k^*$ is bounded by $O_P(k_n)$. That is, $\tilde k_{\A,n}-k^*=O_P(n^{0.5+\epsilon})$ for any $\epsilon>0$.  For more details of some theoretical results on delay time in sequential procedure,  we refer the reader to \cite{aue2004delay} and  \cite{aue2009delay}.
\end{remark}
\subsection{Application of DPD based sequential test to GARCH models}\label{subsec:GARCH}
Consider the following GARCH($p,q$) model:
\begin{eqnarray}\label{GARCH}
\begin{split}
X_t &= \sigma_t(\T)\,\epsilon_t,\\
\sigma_t^2(\T)&= \omega+\sum_{i=1}^p\A_{i} X_{t-i}^2+\sum_{j=1}^q\B_{j}\sigma^2_{t-j}(\T),
\end{split}
\end{eqnarray}
where $\T=(\W,\A_1,\cdots,\A_p,\B_1,\cdots,\B_q)'$ and $\{\epsilon_t | t\in
\mathbb{Z}\}$ is a sequence of i.i.d. random variables with zero mean and unit variance. We assume that
the process $\{X_t | t\in \mathbb{Z}\}$ from the model (\ref{GARCH}) with the true parameter $\T_0$ is strictly stationary
 and ergodic. It is well known that GARCH model has a unique strictly stationary and ergodic solution if and only if the top Lyapunov  exponent is strictly negative. For more details, see, for example, \cite{francq:zakoian:2019}. We further impose the following assumption to ensure the stationarity and ergodicity of $\{\sigma_t^2(\T)|t\in\mathbb{Z}\}$.
 \begin{enumerate}
\item[\bf G1.] $\displaystyle \sup_{\T\in\Theta}\sum_{j=1}^{q}\beta_j <1.$
\end{enumerate}
(cf. \cite{francq:zakoian:2004}).

As a robust estimator for the GARCH model, \cite{lee:song:2009} introduced the following MDPDE:
\begin{eqnarray}\label{MDPDE.GARCH}
 \hat{\theta}_{\A,n}= \argmin_{\theta \in \Theta}\frac{1}{n}\sum_{t=1}^n \tilde{l}_\A(X_t;\T)
 :=\argmin_{\theta \in \Theta} \frac{1}{n}\tilde{H}_{\A,n}(\T),
\end{eqnarray}
 where
\begin{eqnarray*}
\tilde{l}_\A(X_t;\T) &=&
\left\{
\begin{array}{lc}{\displaystyle\Big(\frac{ 1}{\sqrt{{\tilde{\sigma}}_t^2(\T)}}\Big)^\alpha
\Big\{\frac{1}{\sqrt{1+\alpha}}-\Big(1+\frac{1}{\alpha}\Big)
\exp\Big(-\frac{\A}{2}\frac{ X_t^2}{{\tilde{\sigma}}_t^2(\T) }\Big)
\Big\}} &,\A > 0
 \\ \\
 {\displaystyle \frac{X_t^2}{\tilde{\sigma}_t^2(\T)}
+\log\tilde{\sigma}_t^2(\T)}&,\A=0\end{array}\right.
\end{eqnarray*}
and  $\{{\tilde{\sigma}}_t^2(\T)| 1\leq t\leq n\}$ is obtained recursively by
\begin{eqnarray*}\label{tilde1}
{\tilde{\sigma}}_t^2(\T)= \W + \sum_{i=1}^{p}\alpha_i\, X_{t-i}^2+\sum_{j=1}^{q}\beta_j{\tilde{\sigma}}_{t-j}^2(\T).
\end{eqnarray*}
Here, the initial values could be arbitrarily chosen. $l_\A(X_t;\T)$ is defined by replacing $\tilde{\sigma}_t^2(\T)$ in $\tilde{l}_\A(X_t;\T)$ with $\sigma_t^2(\T)$ and it is readily check that $l_\A(X_t;\T)$ is twice continuously differentiable w.r.t $\T$.
For each $\T\in\Theta$ , since $\{\sigma_t^2(\T)|t \in \mathbb{Z}\}$ is strictly stationary and ergodic by assumption {\bf G1},  $\{l_\A(X_t;\T) |t\in \mathbb{Z}\}$ also becomes a strictly stationary ergodic process and thus assumption {\bf M1} in Subsection \ref{Sub3:1} is satisfied. The following assumptions are further made to establish the asymptotics of the MDPDE above.
\begin{enumerate}
\item[\bf G2.] $\theta_0 \in \Theta\,$ and $\Theta$ is a compact subset in $(0,\infty)\times[0,\infty)^{p+q}$.
\item[\bf G3.] If $q>0\,$, ${\mathcal{A}}_{\theta_0}(z)$ and
${\mathcal{B}}_{\theta_0}(z)$ have no common root,
${\mathcal{A}}_{\theta_0}(1) \neq 1$, and $\alpha_{0p}+\beta_{0q}
\neq 0$, where ${\mathcal{A}}_{\T}(z)=\sum_{i=1}^p \A_i\,z^i$
and ${\mathcal{B}}_{\T}(z)=1-\sum_{j=1}^q\beta_j\,z^j$.
(Conventionally, ${\mathcal{A}}_{\T}(z)=0$ if $p=0$ and
${\mathcal{B}}_{\T}(z)=1$ if $q=0$.)
\item[\bf G4.] $\theta_0$ is in the interior of $\Theta$.
\end{enumerate}
The following results are established by \cite{lee:song:2009}, from which and Lemma \ref{inverse} we can see that assumptions {\bf M2} and {\bf M3} hold.
\begin{prop}\label{Lee:Song}
For each $\A \geq 0,$ let $\{ \hat{\T}_{\A,n}\}$ be a sequence of
the MDPDEs satisfying $\eqref{MDPDE.GARCH}$. Suppose that $\epsilon_t$s are i.i.d. random variables from $N(0,1)$.
Then, under assumptions {\bf G1}--{\bf G3}, $\hat{\T}_{\A,n}$ converges
to $\theta_0$ almost surely. If, in addition, assumption {\bf G4} holds, then
\[\sqrt{n}\,(\,\hat{\theta}_{\A,n} -\,\theta_0) \stackrel
{d}{\longrightarrow} N\,\Big(\,0\,,
\mathcal{J}_\A^{-1} \,\mathcal{I}_\A \,\mathcal{J}_\A^{-1}\,\Big)\,,\]
where $\mathcal{J}_\A=\E \big[ \paa l_\A(X_t;\T_0)\big]$ and $\mathcal{I}_\A=\E \big[ \pa_\theta l_\A(X_t;\T_0) \pa_{\theta'} l_\A(X_t;\T_0)\big]$.
\end{prop}
\noindent As a consistent estimator of $\mathcal{I}_\A$,  let us consider\vspace{-0.19cm}
\begin{eqnarray}\label{I.GARCH}
\hat{\mathcal{I}}_{\A,n}=\frac{1}{n}\sum_{t=1}^n \pa_{\theta}
\tilde l_\A(X_t; \hat\T_{\A,n})\,\pa_{\theta'}\tilde{l}_\A(X_t;\hat \T_{\A,n}).
\end{eqnarray}
For the consistency of the estimator, see Lemma 6 in \cite{song:kang:2019}. Then, based on (\ref{I.GARCH}) and $\tilde H_{\A,n}(\T)$
in (\ref{MDPDE.GARCH}), we can construct a detector and stopping time for GARCH models  as in (\ref{detector.ts}) and (\ref{stop.ts}), respectively.

We now check whether the conditions introduced in Subsections \ref{Sub3:2} are fulfilled.
Under assumptions {\bf G1}--{\bf G4}, conditions {\bf S1}--{\bf S3} were already shown in \cite{lee:song:2009} when proving the asymptotic normality of the MDPDE. For details, see (i) and (iv) on page 337 of \cite{lee:song:2009}. Condition {\bf S4} can be found in Lemma 3 in \cite{song:kang:2019}. Hence, the result in Theorem \ref{thm4} holds for the GARCH models with Gaussian errors.

For the consistency of the test procedure, we assume that assumptions {\bf G1}--{\bf G4} are also fulfilled for $\T_1(\neq \T_0)$. Then, similarly to the above, assumptions {\bf M1}--{\bf M3} are also satisfied  for $\T_1$. Since condition {\bf S7} is a generally accepted assumption when $\T_1\neq\T_0$, we just deal with  conditions {\bf S5} and {\bf S6}. For this, we introduce a further moment condition, that is, $\E [(X_t^{\T_1})^{4+\epsilon}] <\infty$ for some $\epsilon>0$. For $\T_1$ satisfying this moment condition, Lemma \ref{S56} shows that \vspace{-0.19cm}
\[ \sum_{t>n+k^*}^{n+k_n} \sup_{\T\in \widetilde \Theta} \big\| \pa_{\theta}\,l_\A(X^{\T_1}_t;\theta)-\pa_{\theta}\,\tilde l_\A(X^{\T_1}_t;\theta)\big\|= O(1)\quad a.s.,\]
where we note that  $\widetilde \Theta$ can be taken as any subset of $\Theta$ whose elements have components $\A_i$ and $\B_j$ bounded away from zero. Since $\T_0$ is in the interior of $\Theta$, one can take a neighborhood $N_3(\T_0)$ satisfying conditions {\bf S5} and {\bf S6}. Hence, the consistency of the test procedure is asserted.
\begin{thm}\label{thm6}
Let $\tilde k_{\A,n}$ be the stopping time in (\ref{stop.ts}) that is constructed using (\ref{I.GARCH}) and $\tilde H_{\A,n}(\T)$ in (\ref{MDPDE.GARCH}). Suppose that assumptions {\bf G1}-{\bf G4} hold for $\T_0$ and $\T_1(\neq\T_0)$. If the boundary function $b(\cdot)$ satisfies assumption {\bf B}, we have that under $H_0$,
\begin{eqnarray*}
\lim_{n\rightarrow\infty} P\big( \tilde k_{\A,n} <\infty\ |\ H_0 \big)=P\left(  \sup_{0<s<1} \frac{ \big\| W_d(s)\big\|}{b(s/(1-s))} >1 \right),
\end{eqnarray*}
where $\{W_d(s)\}$ is a  $d$-dimensional standard Wiener process and $d=p+q+1$.  In addition, if $\E [(X_t^{\T_1})^{4+\epsilon}] <\infty$ for some $\epsilon>0$ and  $\sup_{t>0} b(t) <\infty$, then under $H_1$,
\begin{eqnarray*}
\lim_{n\rightarrow\infty} P\big( \tilde k_{\A,n} <\infty\ |\ H_1 \big)=1.
\end{eqnarray*}
\end{thm}
\begin{remark}\label{berkes}
The detector $\tilde D_{\A=0,n}(k)$ in the GARCH models above  is essentially equal to the score based detector of \cite{berkes.et.al:2004seq} that is  constructed based on the quasi-MLE of \cite{berkes.et.al:2003}. Their sequential procedure has been recently extended to ARMA-GARCH models by \cite{song:kang:2020}.
\end{remark}

\section{Simulation study}\label{Sec:sim}
In the present section, we evaluate the finite sample performance of the proposed sequential procedure and compare with the score based procedure in the following GARCH(1,1) models:
\begin{eqnarray*}
\begin{split}
X_t &= \sigma_t(\T)\,\epsilon_t,\\
\sigma_t^2(\T)&= \W+\A_1 X_{t-1}^2+\B_1\sigma^2_{t-1}(\T),
\end{split}
\end{eqnarray*}
where $\{\epsilon_t \}$ is a sequence of i.i.d. random variables from $N(0,1)$ and $\T=(\W,\A_1,\B_1)$.
We use the maximum norm  and  the constant boundary function to implement the test procedures at the significance level of 5\%. The corresponding critical value for GARCH(1,1) model is 2.632, which is given in Table \ref{tab:cr.value}. The sample sizes of the historical data under consideration are $n=$500, 1000, and 1500, and the detector $\tilde D_{\A,n}(k)$ is monitored until $t=n+2000$. If the detector $\tilde D_{\A,n}(k)$ crosses over the critical value during the monitoring period, $H_0$ is rejected. The empirical sizes and powers are calculated as the ratio of the rejection number out of 2000 repetitions.  We consider $\A$ values in $\{0, 0.1, 0.2, 0.3, 0.5\}$. As mentioned in Remark \ref{berkes}, $\tilde D_{\A,n}(k)$ with $\A=0$ represents the score based detector. We shall check the robustness of the proposed procedure using $\tilde D_{\A,n}(k)$ with $\A>0$ and compare with the procedure with $\A=0$.\\

\noindent {\bf (i) Simulation results for uncontaminated cases}

We first address the cases where data are not contaminated by outliers. The following two parameters are considered to evaluate the empirical sizes:\vspace{-0.2cm}
\[ \T_1=(0.2, 0.3, 0.2),\quad \T_2=(0.2, 0.1, 0.8), \vspace{-0.2cm}\]
 where the second one is employed to see the performance in a more volatile situation. The empirical sizes are depicted in  Figure \ref{fig:size}, where $k$ denotes the monitoring time after $n$. We first note that some oversizes are observed particularly when $n$=500 and $k$=2000
 \begin{figure}[!t]
\includegraphics[height=0.6\textwidth,width=1.0\textwidth]{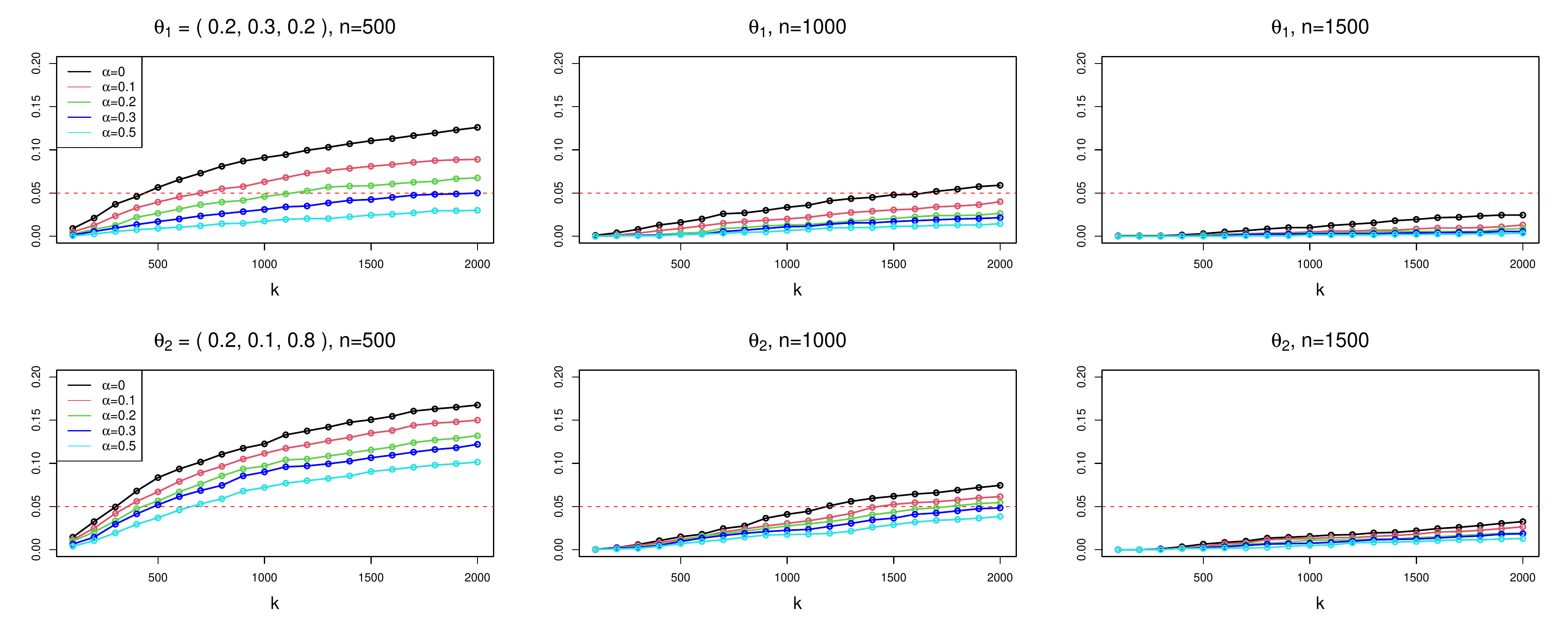}\vspace{-1cm}
\caption{\small The plots of the empirical sizes when no outlier exists. The dashed red lines represent the significance level of  5\%.} \label{fig:size}
\end{figure}
but all the empirical sizes become less than the nominal level 0.05 as $n$ increases.
This means that when no outlier exists, both of the score based and the proposed procedures perform properly under $H_0$. Additionally,  we can see that the empirical sizes decrease with an increase in $\A$, and a little higher sizes are observed  in the case of the parameter $\T_2$.

\begin{figure}[!t]
\includegraphics[height=0.9\textwidth,width=1.0\textwidth]{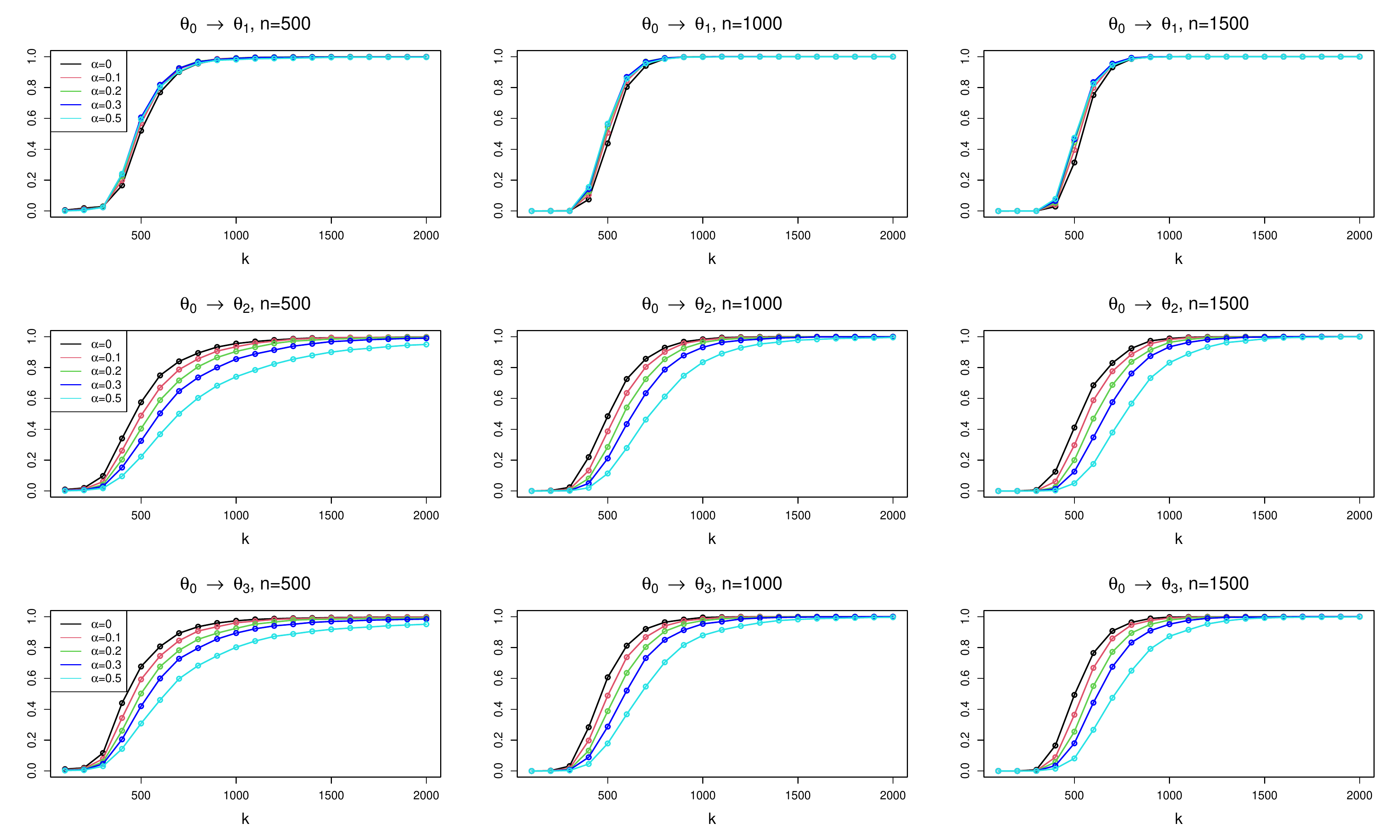}\vspace{-1cm}
\caption{\small The plots of the empirical powers   when the parameter changes at $k^*=250$  and no outlier exists.}\label{fig:power}\vspace{0.8cm}
\includegraphics[height=0.3\textwidth,width=1.0\textwidth]{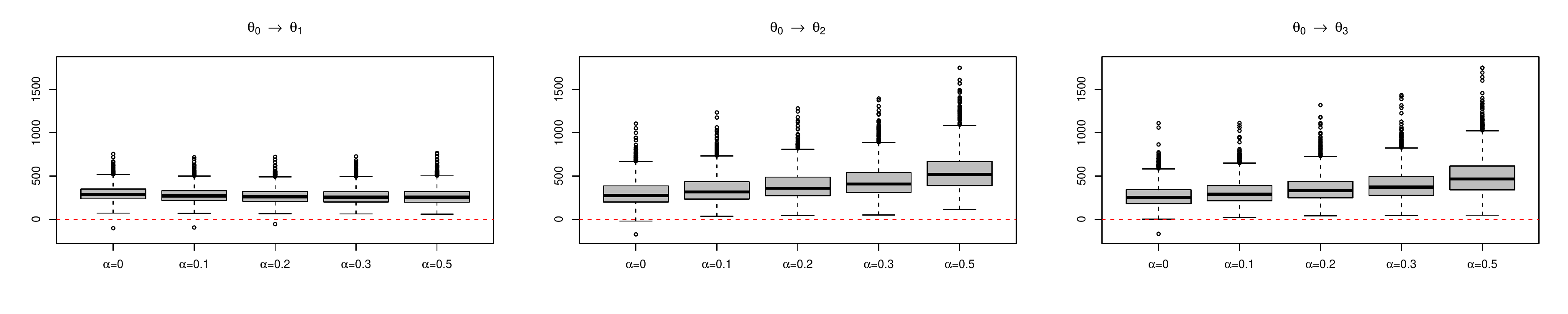}\vspace{-1cm}
\caption{\small The box plots of the delay times  for $n=1500$ when no outlier exists.}\label{fig:delay}
\end{figure}

To evaluate the empirical powers, we set the parameter $\T_0:=(0.2,0.2,0.6)$ as the default parameter and then change to  $\T_1$,  $\T_2$,  and $\T_3$=(0.5, 0.2, 0.6) at $k^*=250$.  The first change  s  a change to a less  volatile state,  which is set to be relatively larger than the second change  to $\T_2$.  The third one reflects the level shift of the volatility without any change in  $\A_1$ and $\B_1$.
The results are presented in Figure \ref{fig:power}.  One can see that all the procedures yields reasonably good powers approaching one after the change point $k^*$=250.  In particular,  the empirical powers of the score based procedure approaches more rapidly.  As expected,  it is observed that the power of the proposed procedure decreases generally  with an increase in  $\A$.  It should  be  noted that although the powers of the procedure with a large $\A$  such as $\A=0.5$ get closer to one as $k$ increases,  the powers evaluated near after the change point are quite low compared to the procedure with $\A=0$ or $\A=0.1$.   See,  for example,  the values at $k=500$  in the middle and bottom panels in Figure \ref{fig:power}.   This indicates that too large $\A$ can decreases the efficiency of the proposed procedure under $H_1$.  When the change is comparatively large,  all the procedures perform similarly,  see the top panel.
 Our findings can also be seen  in the box plots of the delay times, i.e., $\tilde k^\A_n -k^*,$ displayed in Figure \ref{fig:delay},  where we can see the general tendency that the smaller $\A$, the shorter the detection delay tends to be.  Table \ref{tab:delay} presents the  average delay times and  its ratios  to that of the score based procedure.   The delay time tends to be longer as the sample size of the historical data $n$ increases.  The procedure with $\A$ in [0.1, 0.3] is observed to take  about 1.1 to 1.5 times  longer than the score based one in the second and the third changes.

Overall, our proposed procedure performs reasonably in the case of no outlying observations. The score based procedure generally outperforms the proposed procedure and the performance of our procedure decreases as $\A$ increases.  It should also be noted that  a large tuning parameter $\A$ can lead to some   loss in efficiency and thus,  based on the above results,  we first recommend practitioner to use $\A$ less than 0.5 when data appears not to be contaminated or the degree of the contamination does not seem to be severe.

\begin{table}[htbp]
    \caption{\small Average delay times of the score based procedure($\A=0$) and the proposed procedure ($\A>0$).}
   \tabcolsep=4.2pt
     \renewcommand{\arraystretch}{1.2}
   {\scriptsize
   \centering
  \begin{tabular}{rccccccccccccccccc}
\toprule
        & \multicolumn{5}{c}{$\theta_0 \rightarrow \theta_1$} &       & \multicolumn{5}{c}{$\theta_0 \rightarrow \theta_2$} &       & \multicolumn{5}{c}{$\theta_0 \rightarrow \theta_3$} \\
\cmidrule{2-6}\cmidrule{8-12}\cmidrule{14-18}      & $\A$  &       &       &       &       &       & $\A$  &       &       &       &       &       & $\A$  &       &       &       &  \\
\multicolumn{1}{c}{$n$} & 0     & 0.1   & 0.2   & 0.3   & 0.5   &       & 0     & 0.1   & 0.2   & 0.3   & 0.5   &       & 0     & 0.1   & 0.2   & 0.3   & 0.5 \\
\cmidrule{2-6}\cmidrule{8-12}\cmidrule{14-18}\multicolumn{1}{c}{500} & 245   & 232   & 222   & 218   & 219   &       & 216   & 256   & 303   & 349   & 450   &       & 171   & 210   & 250   & 290   & 382 \\
      & [1.00] & [0.95] & [0.91] & [0.89] & [0.89] &       & [1.00] & [1.18] & [1.40] & [1.61] & [2.08] &       & [1.00] & [1.23] & [1.46] & [1.70] & [2.23] \\
\multicolumn{1}{c}{1000} & 266   & 249   & 240   & 235   & 233   &       & 255   & 293   & 336   & 380   & 473   &       & 215   & 255   & 295   & 338   & 422 \\
      & [1.00] & [0.94] & [0.90] & [0.88] & [0.88] &       & [1.00] & [1.15] & [1.32] & [1.49] & [1.85] &       & [1.00] & [1.19] & [1.37] & [1.57] & [1.96] \\
\multicolumn{1}{c}{1500} & 287   & 270   & 260   & 257   & 255   &       & 277   & 316   & 360   & 408   & 515   &       & 252   & 291   & 330   & 373   & 466 \\
      & [1.00] & [0.94] & [0.91] & [0.90] & [0.89] &       & [1.00] & [1.14] & [1.30] & [1.47] & [1.86] &       & [1.00] & [1.15] & [1.31] & [1.48] & [1.85] \\
\bottomrule\\ \vspace{-0.6cm}
\end{tabular}%
      }
      {\small The figures in the brackets are the ratios of the average delay time to that of the score based procedure.}
  \label{tab:delay}%
\end{table}%

\newpage
\noindent {\bf (ii) Simulation results for contaminated cases}

 Next, we explore the performances in the contaminated cases. For this, we consider the following three scenarios: (i) outliers are observed in the historical data (say, H-type); (ii) outliers occur during the monitoring period (M-type); (iii) outliers are involved in both of the historical data and the monitoring period (HM-type). In the case (i), we generate the historical data $\{X_1,\cdots,X_n\}$ as the following scheme: $X_t=X_{t,o} + s\cdot p_t \cdot sign(X_{t,o})$, where $\{X_{t,o}\}$ is the uncontaminated path
\begin{figure}[!t]
\includegraphics[height=0.6\textwidth,width=1.0\textwidth]{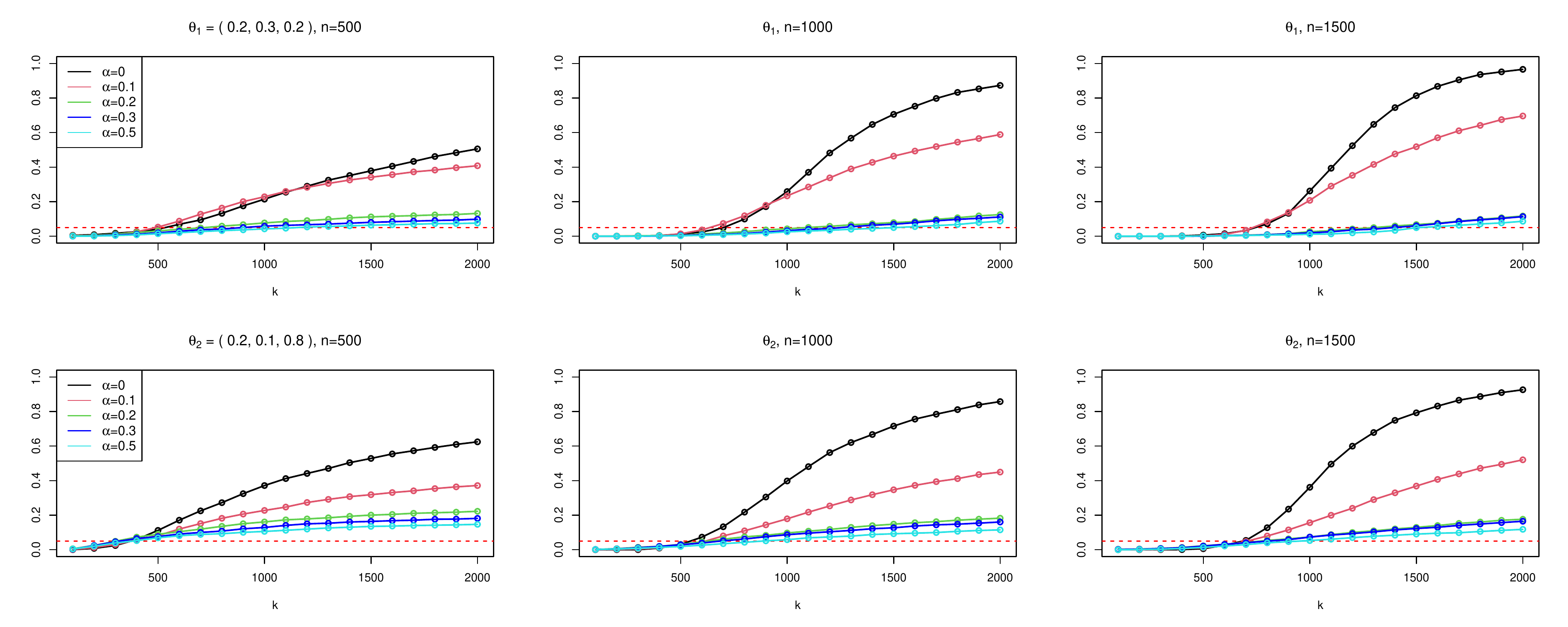}\vspace{-1cm}
\caption{\small The plots of the empirical sizes when  outliers are included in the historical data (H-type case).}\label{fig:size_h}\vspace{0.3cm}
\includegraphics[height=0.6\textwidth,width=1.0\textwidth]{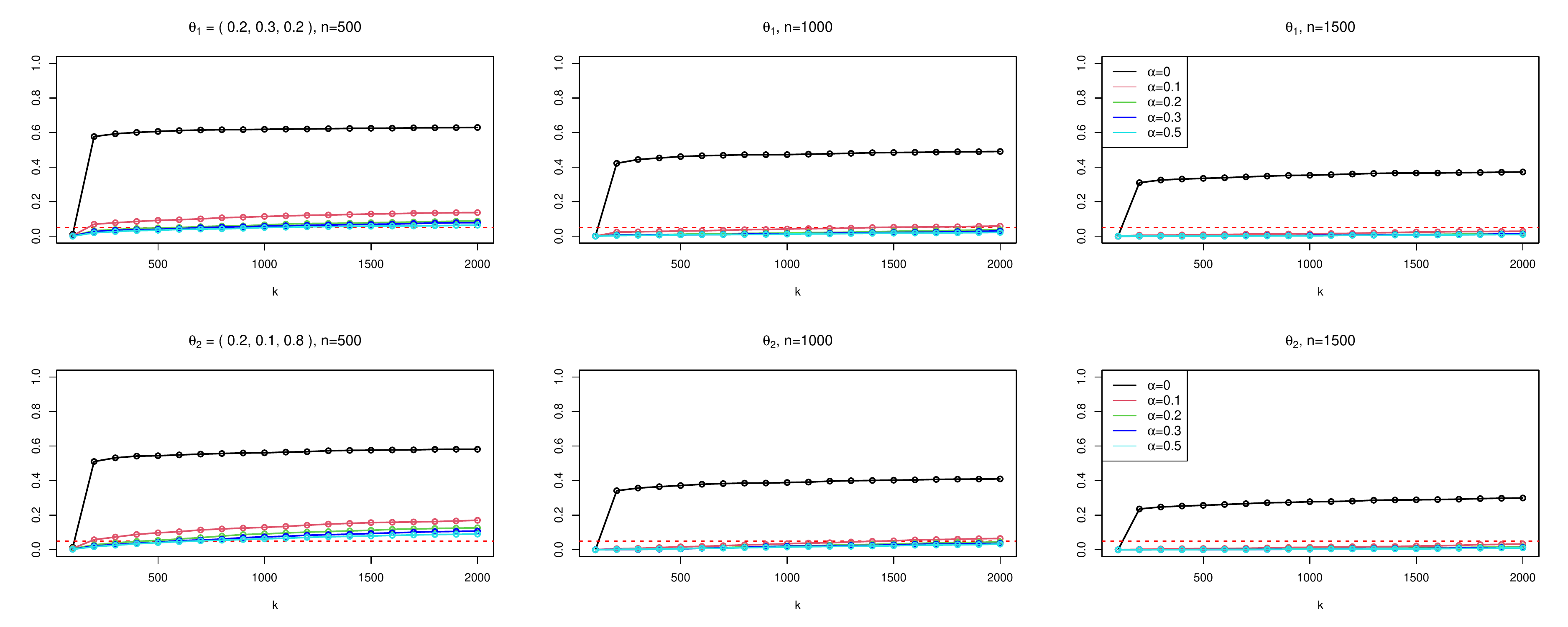}\vspace{-1cm}
\caption{\small The plots of the empirical sizes when  outliers are included in the arrived data (M-type case).}\label{fig:size_m}
\end{figure}
 from the GARCH(1,1) model above,  $s=5\sqrt{\W/(1-\A_1-\B_1)}$, i.e., 5 times the standard deviation of  the GARCH process, and $\{p_t\}$ is a sequence of i.i.d. random variables from Bernoulli distribution with parameter $p$. In this simulation, $p=0.03$ is considered. $\{X_{t,o}\}$ and $\{p_t\}$ are assumed to be independent. For the case (ii), we use the same method to contaminate $\{X_{n+1},\cdots, X_{n+200}\}$, which describes the situation that outliers occurs before the change point $k^*=250$. In the case (iii), we contaminate both of the historical data and $\{X_{n+1},\cdots,X_{n+200}\}$ using the aforementioned method.

\begin{figure}[!t]
\includegraphics[height=0.6\textwidth,width=1.0\textwidth]{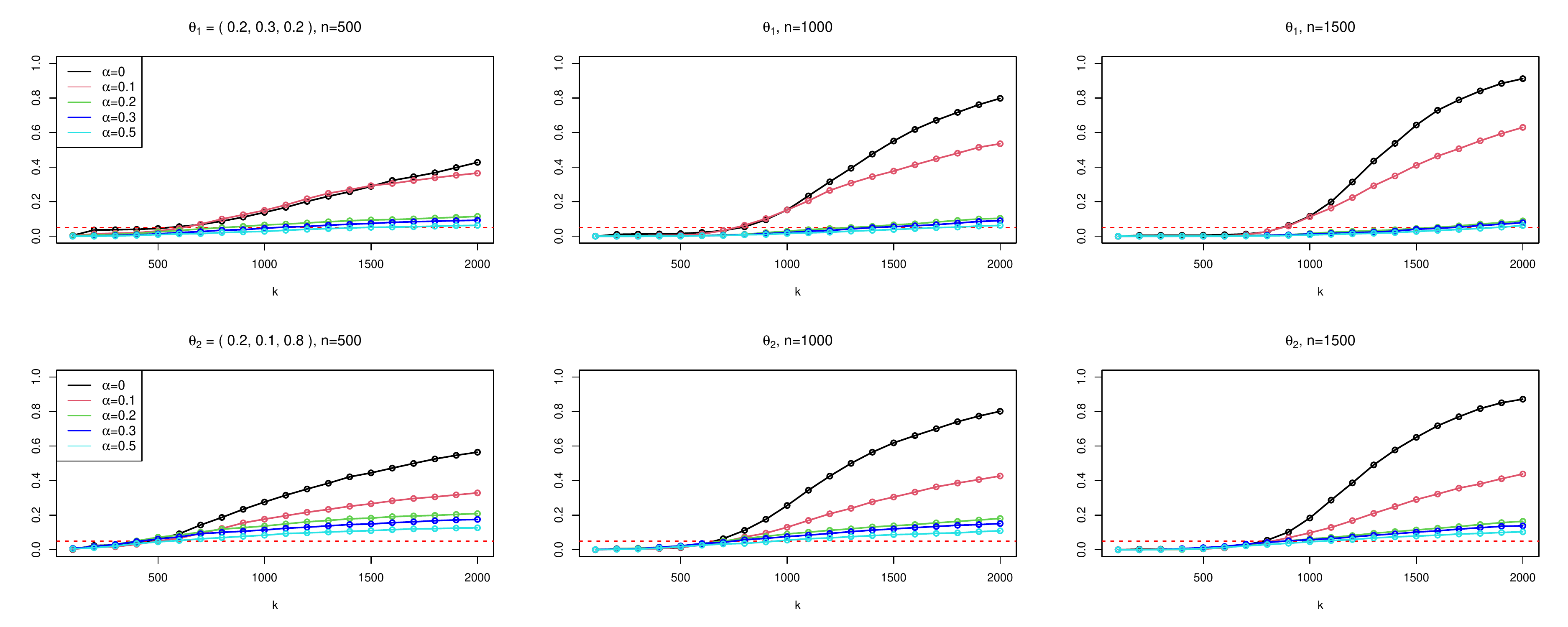}\vspace{-0.75cm}
\caption{\small The plots of the empirical sizes when  outliers are included in both  the historical and the arrived data (HM-type case).}\label{fig:size_hm}
\end{figure}

The empirical sizes for  the three contaminated cases are given in Figures \ref{fig:size_h} - \ref{fig:size_hm}.  As shown in the figures, the score based procedure exhibits severe size distortions for all three scenarios.  On the other hand, the proposed procedure, in particular  with $\A\geq 0.2$, yields fairly stable sizes. The procedure with $\A=0.1$ is observed to be somewhat distorted by outliers except for M-type case. As can be seen in Figure \ref{fig:size_m}, the score based procedure is distorted near the point where the outliers occur whereas our procedure with $\A>0$ are hardly affected. From this, it can be surmised that the proposed procedure is more robust to outliers occurring in the monitoring period. Interestingly, the results of the third scenario shown in Figure \ref{fig:size_hm} are similar to the first case in Figure \ref{fig:size_h}. The procedure damaged by the outliers in the historical data seems to  become insensitive to the newly occurring outliers.

The empirical powers and the corresponding box plots of the delay times under H-, M-, and HM- type contaminations are  presented in Figures \ref{fig:power_H}-\ref{fig:delay_HM}, respectively.  From Figures \ref{fig:power_H} and  \ref{fig:power_HM}, we can observe the significant power losses of the score based procedure whereas the procedure with $\A\geq0.2$ shows the power curves similar to the ones in Figure \ref{fig:power}, i.e, obtained in the uncontaminated cases. This indicates that the proposed procedure with $\A \geq0.2$ is adequately robust against the outliers. Although some power loses are observed for smaller $\A=0.1$ in the first and third scenarios,  for  all $\A$ values considered, our procedure shows strong robustness to outliers in the arrived data, see Figure \ref{fig:power_M}, where the higher powers of the score based procedure near $k$=200 reflect that the score based procedure is compromised by the outliers in the monitoring period. The box plots in Figures \ref{fig:delay_H} and \ref{fig:delay_HM} also consistently show the distortions of the score based procedure and the robustness of our procedure with $\A\geq 0.2$.   The early detections of the score based procedure in  Figure \ref{fig:delay_M} are due to the outliers that  occur before the change point.  This phenomenon is clearly observed particularly when $n$ is small.  For example,  in Table \ref{tab:delay_outliers},  we can see minus average delay times in the case of M-type outlier and  $n=500.$  Here   $d_\A$ in the brackets is  defined as follows:
\[ d_\A :=\frac{\mbox{ Average delay time of the procedure with } \A \mbox { in contaminated case}}{\mbox{ Average delay time of the procedure with } \A \mbox { in uncontaminated case}}.\]
A value of  $d_\A$ closer  to one means that the procedure is less affected  by the corresponding contamination.  It can be seen in the table that $d_\A$ with $\A \geq 0.2$ is  much closer to one than $d_{\A=0}$.  
 
Overall, our findings strongly support the validity and robustness of the proposed sequential procedure.  In particular,  our procedure with a small $\A$ is sufficiently robust against outliers and powerful as the score based procedure.  For the reasons mentioned in Remark \ref{choice.A},  we do not select an optimal $\A$,  but  we recommend to consider  $\A$ values in $[0.2,0.3]$ based on our simulation results.


\begin{figure}
\includegraphics[height=0.9\textwidth,width=1.0\textwidth]{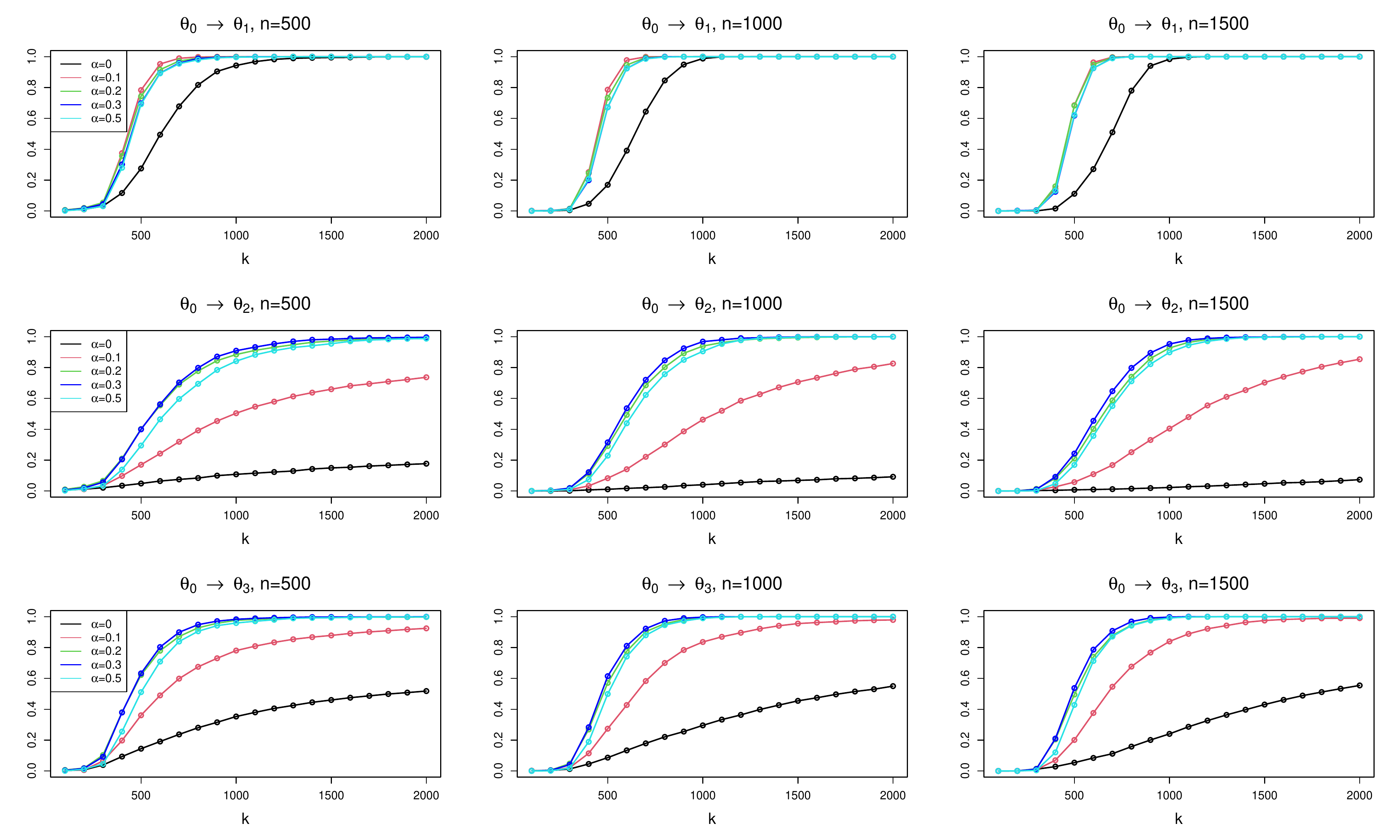}\vspace{-0.75cm}
\caption{\small The plots of the empirical powers when the parameter changes at $k^*=250$ and outliers are included in the historical data (H-type case).}\label{fig:power_H}\vspace{0.8cm}
\includegraphics[height=0.3\textwidth,width=1.0\textwidth]{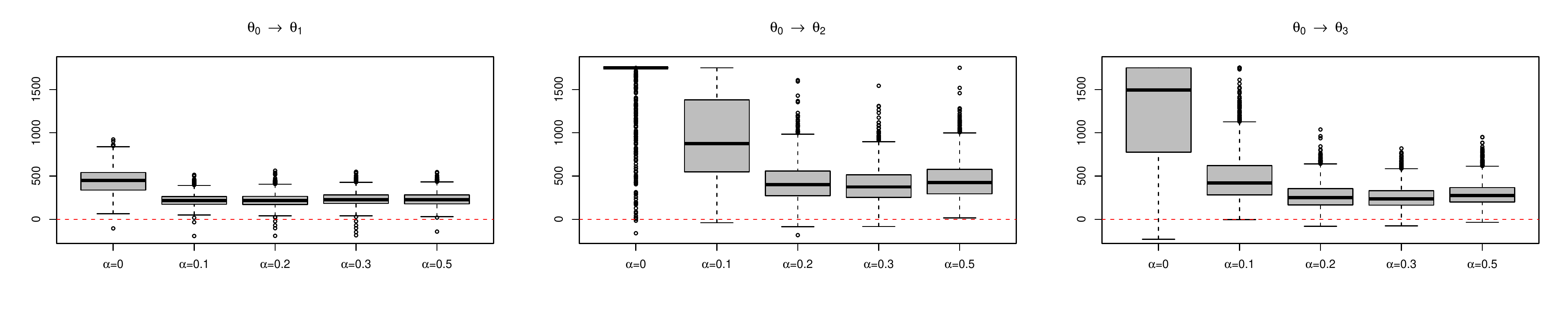}\vspace{-1cm}
\caption{\small The box plots of the delay times  for $n=1500$ when  outliers are included in the historical data.}\label{fig:delay_H}
\vspace{0.8cm}
\end{figure}

\begin{figure}
\includegraphics[height=0.9\textwidth,width=1.0\textwidth]{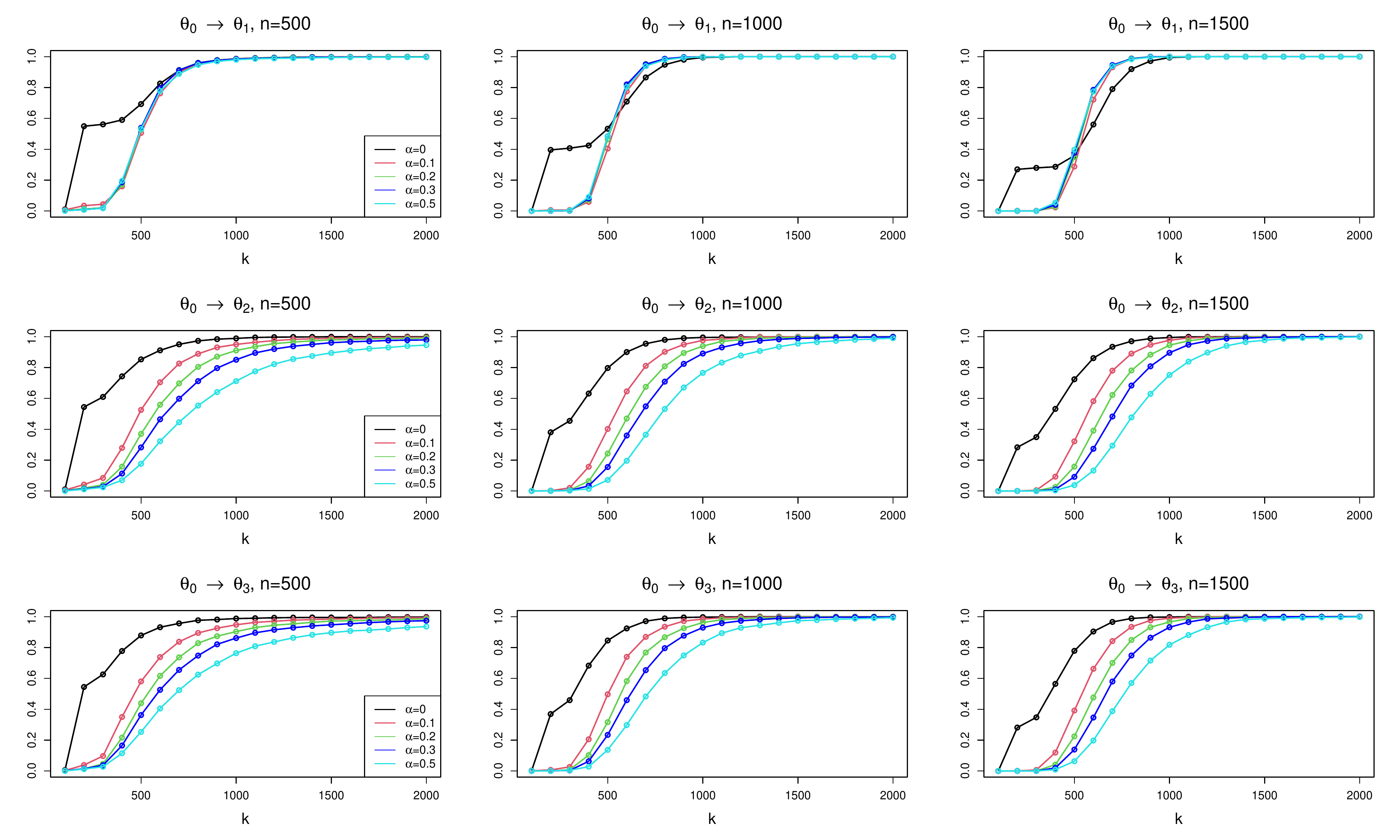}\vspace{-0.75cm}
\caption{\small The plots of the empirical powers when the parameter changes at $k^*=250$ and outliers are included in the arrived data (M-type case).}\label{fig:power_M}\vspace{0.8cm}
\includegraphics[height=0.3\textwidth,width=1.0\textwidth]{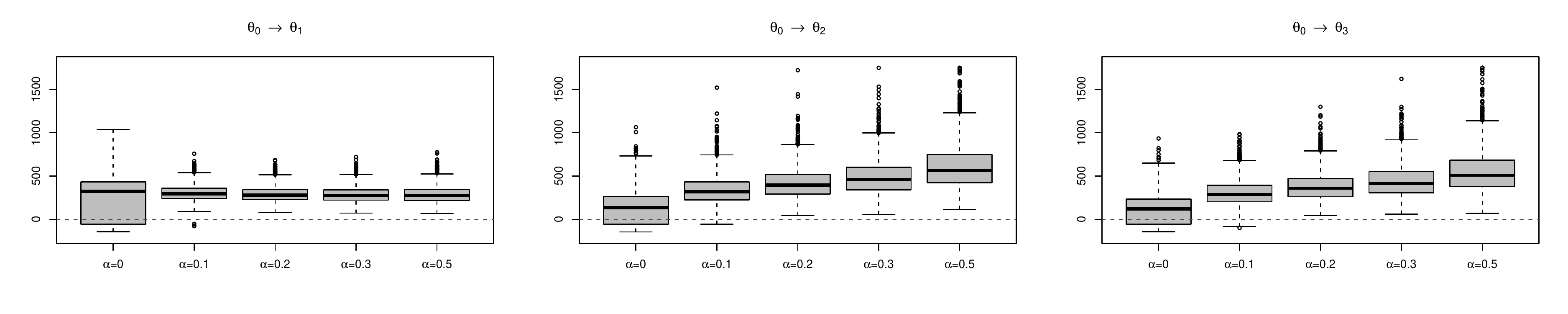}\vspace{-1cm}
\caption{\small The box plots of the delay times  for $n=1500$ when  outliers are included in the arrived data.}\label{fig:delay_M}
\vspace{0.8cm}
\end{figure}

\begin{figure}
\includegraphics[height=0.9\textwidth,width=1.0\textwidth]{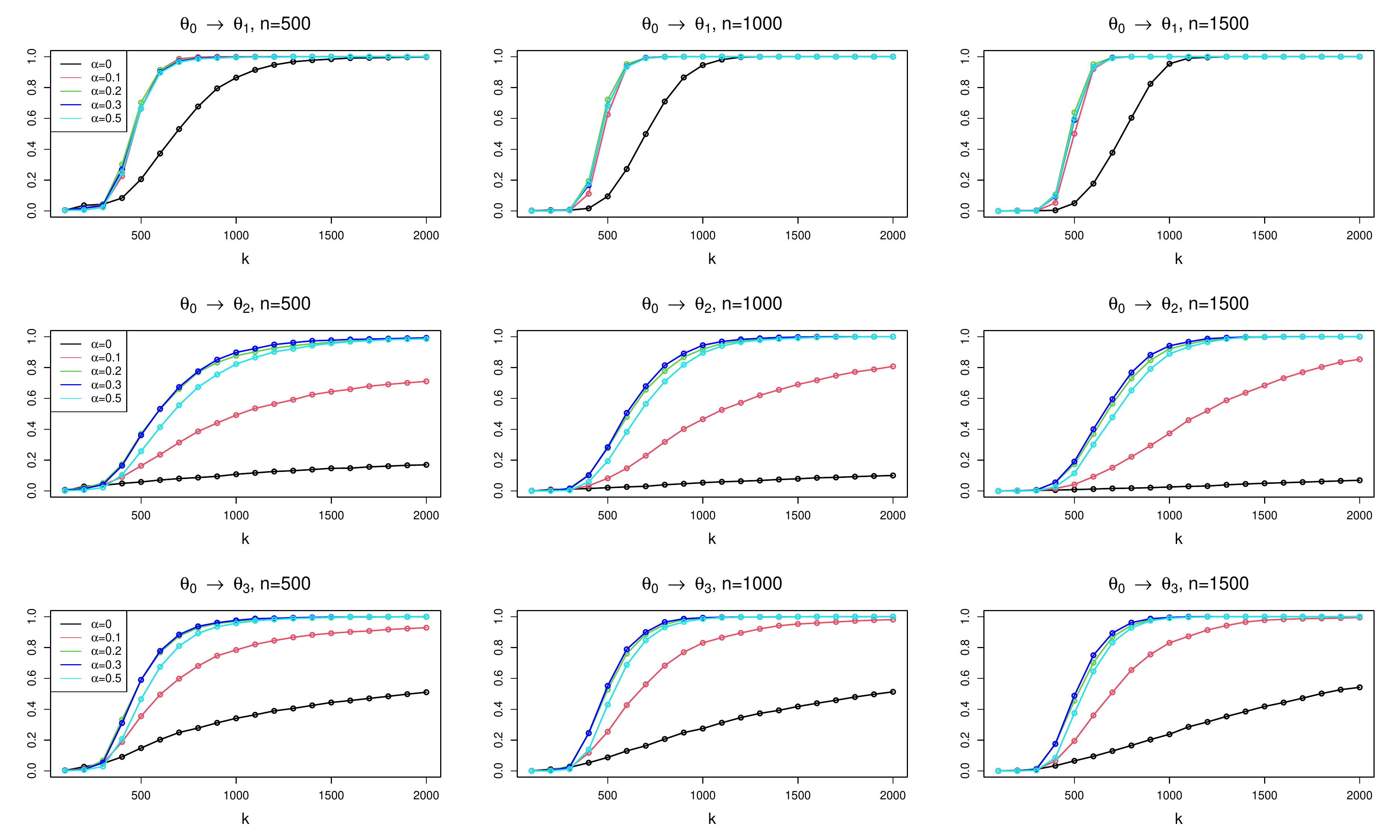}\vspace{-0.75cm}
\caption{\small The plots of the empirical powers when the parameter changes at $k^*=250$ and outliers are included in both the historical  and arrived data (HM-type case).}\label{fig:power_HM}\vspace{0.8cm}
\includegraphics[height=0.3\textwidth,width=1.0\textwidth]{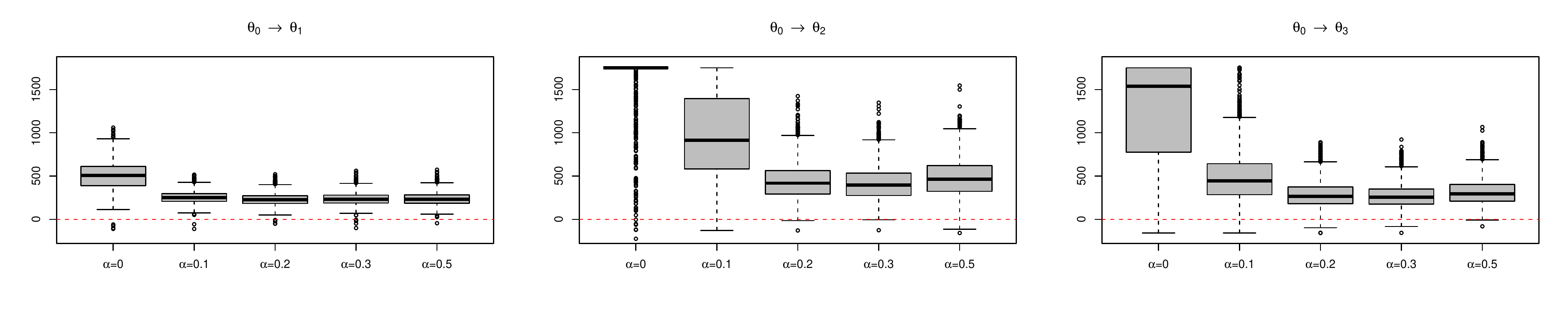}\vspace{-1cm}
\caption{\small The box plots of the delay times  for $n=1500$ when  outliers are included in both the historical and arrived  data.}\label{fig:delay_HM}
\vspace{0.8cm}
\end{figure}

\begin{table}[h]
  \caption{\small Average delay times of the score based procedure ($\A=0$) and the proposed procedure ($\A>0$) in the presence of outliers.}
   \tabcolsep=3pt
     \renewcommand{\arraystretch}{1.2}
   {\scriptsize
   \centering
   \begin{tabular}{crrccccccccccccccccc}
\toprule
      &       &       & \multicolumn{5}{c}{$\theta_0 \rightarrow \theta_1$} &       & \multicolumn{5}{c}{$\theta_0 \rightarrow \theta_2$} &       & \multicolumn{5}{c}{$\theta_0 \rightarrow \theta_3$} \\
\cmidrule{4-8}\cmidrule{10-14}\cmidrule{16-20}\multicolumn{1}{p{3.165em}}{Outlier} &       &       & $\A$  &       &       &       &       &       & $\A$  &       &       &       &       &       & $\A$  &       &       &       &  \\
\multicolumn{1}{l}{Type} & \multicolumn{1}{c}{$n$} &       & 0     & 0.1   & 0.2   & 0.3   & 0.5   &       & 0     & 0.1   & 0.2   & 0.3   & 0.5   &       & 0     & 0.1   & 0.2   & 0.3   & 0.5 \\
\cmidrule{1-2}\cmidrule{4-8}\cmidrule{10-14}\cmidrule{16-20}      & \multicolumn{1}{c}{500} &       & 353   & 176   & 183   & 198   & 200   &       & 1751  & 744   & 310   & 306   & 377   &       & 1558  & 360   & 193   & 190   & 244 \\
      &       &       & [1.44] & [0.76] & [0.82] & [0.91] & [0.91] &       & [8.11] & [2.91] & [1.02] & [0.88] & [0.84] &       & [9.11] & [1.71] & [0.77] & [0.66] & [0.64] \\
H     & \multicolumn{1}{c}{1000} &       & 394   & 195   & 199   & 212   & 215   &       & 1751  & 816   & 353   & 333   & 384   &       & 1460  & 394   & 222   & 214   & 251 \\
      &       &       & [1.48] & [0.78] & [0.83] & [0.90] & [0.92] &       & [6.87] & [2.78] & [1.05] & [0.88] & [0.81] &       & [6.79] & [1.55] & [0.75] & [0.63] & [0.59] \\
      & \multicolumn{1}{c}{1500} &       & 448   & 217   & 216   & 228   & 227   &       & 1751  & 874   & 402   & 374   & 426   &       & 1496  & 420   & 252   & 238   & 274 \\
      &       &       & [1.56] & [0.80] & [0.83] & [0.89] & [0.89] &       & [6.32] & [2.77] & [1.12] & [0.92] & [0.83] &       & [5.94] & [1.44] & [0.76] & [0.64] & [0.59] \\
\cmidrule{1-2}\cmidrule{4-8}\cmidrule{10-14}\cmidrule{16-20}      & \multicolumn{1}{c}{500} &       & -60   & 249   & 241   & 239   & 240   &       & -59   & 238   & 316   & 380   & 498   &       & -59   & 212   & 282   & 330   & 422 \\
      &       &       & [-0.24] & [1.07] & [1.09] & [1.10] & [1.10] &       & [-0.27] & [0.93] & [1.04] & [1.09] & [1.11] &       & [-0.35] & [1.01] & [1.13] & [1.14] & [1.10] \\
M     & \multicolumn{1}{c}{1000} &       & 228   & 272   & 258   & 255   & 255   &       & 80    & 289   & 365   & 424   & 531   &       & 74    & 252   & 318   & 368   & 461 \\
      &       &       & [0.86] & [1.09] & [1.08] & [1.09] & [1.09] &       & [0.31] & [0.99] & [1.09] & [1.12] & [1.12] &       & [0.34] & [0.99] & [1.08] & [1.09] & [1.09] \\
      & \multicolumn{1}{c}{1500} &       & 324   & 295   & 280   & 274   & 274   &       & 136   & 318   & 397   & 460   & 564   &       & 121   & 287   & 360   & 414   & 510 \\
      &       &       & [1.13] & [1.09] & [1.08] & [1.07] & [1.07] &       & [0.49] & [1.01] & [1.10] & [1.13] & [1.10] &       & [0.48] & [0.99] & [1.09] & [1.11] & [1.09] \\
\cmidrule{1-2}\cmidrule{4-8}\cmidrule{10-14}\cmidrule{16-20}      & \multicolumn{1}{c}{500} &       & 429   & 213   & 198   & 206   & 207   &       & 1751  & 772   & 333   & 333   & 409   &       & 1674  & 356   & 215   & 213   & 266 \\
      &       &       & [1.75] & [0.92] & [0.89] & [0.94] & [0.95] &       & [8.11] & [3.02] & [1.10] & [0.95] & [0.91] &       & [9.79] & [1.70] & [0.86] & [0.73] & [0.70] \\
      & \multicolumn{1}{c}{1000} &       & 452   & 226   & 205   & 212   & 213   &       & 1751  & 805   & 360   & 348   & 408   &       & 1665  & 404   & 240   & 232   & 274 \\
HM    &       &       & [1.7] & [0.91] & [0.85] & [0.9] & [0.91] &       & [6.87] & [2.75] & [1.07] & [0.92] & [0.86] &       & [7.74] & [1.58] & [0.81] & [0.69] & [0.65] \\
      & \multicolumn{1}{c}{1500} &       & 507   & 250   & 226   & 233   & 232   &       & 1751  & 915   & 418   & 397   & 465   &       & 1540  & 445   & 265   & 254   & 293 \\
      &       &       & [1.77] & [0.93] & [0.87] & [0.91] & [0.91] &       & [6.32] & [2.90] & [1.16] & [0.97] & [0.90] &       & [6.11] & [1.53] & [0.80] & [0.68] & [0.63] \\
\bottomrule\\ \vspace{-0.6cm}
\end{tabular}%
}    

{\small The figures in the brackets are the ratio $d_\A$.}
  \label{tab:delay_outliers}%
\end{table}%

\section{Real data analysis}\label{Sec:real}

This section presents two real data applications. The first one is given to describe that our sequential procedure works as well as the score based procedure when there is no outlying observation. In the second application, the data  including  some deviating observations are considered and we demonstrate that the score based procedure could result in a misleading result whereas the proposed procedure produces reliable result.

We analyse the log return series of the S\&P500 index from Jan 2000 to Dec 2004 and the Hang Seng index from Jan 1988 to Dec 1996. We fit  the GARCH(1,1) model  to each data set since each series shows typical features such as arch effect and it is also the most commonly used model in empirical practice.
Considering the possibility of outliers being present in the historical data, we use a robust test to check whether a parameter change exists in the historical data. In this study, we employ the following robust test for parameter change  introduced by \cite{song:kang:2019}:
\begin{eqnarray*}
\tilde T^\A_n=\max_{1\leq k \leq n}
\frac{1}{n}\pa_{\T'}\tilde{H}_{\A,k}(\HT) \hat{\mathcal{I}}_{\A}^{-1}\,\pa_{\T}\tilde{H}_{\A,k}(\HT):=\max_{1\leq k \leq n} \tilde T^\A_n(k),
\end{eqnarray*}
where $\A\geq0$ is a tuning parameter controlling the trade-off between efficiency and robustness as in MDPDE and
 $\hat{\mathcal{I}}_\A$ is a consistent estimator of $\mathcal{I}_\A=\E \big[ \pa_\theta l_\A(X_t;\T_0) \pa_{\theta'} l_\A(X_t;\T_0)\big]$ usually given by (\ref{I.GARCH}).  Under the null hypothesis of no parameter change, $\tilde T^{\A}_n$  converges in distribution to $\sup_{0\leq s \leq 1} \big\|W_d^o(s)\big\|_2^2$, where $ \{ W_d^o(s) |s\geq0\}$ is the $d$-dimensional standard Wiener process and $d$ is the number of parameters. When $\tilde T^\A_n$ is large, $H_0$ is rejected and the change point is located as $\argmax_{1\leq k\leq n}\tilde T^\A_n(k)$.  We also note that $\tilde T_n^\alpha$ with $\A=0$ becomes the score test for parameter change.  As demonstrated in \cite{song:kang:2019}, $\tilde T_n^\A$ with $\A$ close to zero can effectively detect changes in parameters when seemingly outliers are included in a data set being suspected of having parameter changes.

As in  the simulation study above, we employ again the maximum norm and the constant boundary function to conduct the monitoring  procedure. Here, we also consider $\A$ values in $\{$0, 0.1, 0.2, 0.3, 0.5$\}$. As mentioned in Remark \ref{choice.A}, we do not try to select an optimal $\A$. In stead, we implement the sequential procedure with  each $\A$ and incorporate the results to make a decision. The log return series and the plot of the detector $\tilde D_{\A,n}(k)$ for the S\&P500 index and the Hang Seng index are presented in Figures \ref{SP:2000} and \ref{HS:1988}, respectively, where the dashed blue vertical line in left subfigure denotes the monitoring start date and the dashed horizontal line in the right subfigure indicates the critical value, 2.381, corresponding to the significance level of 10\%.  More detailed analyses are as follows. \vspace{0.3cm}

\begin{table}[!t]
\centering
{\small
  \tabcolsep=5pt
  \renewcommand{\arraystretch}{1.2}
    \caption{Results of the parameter constancy test for the historical data} \label{Tab:alpha}
\begin{tabular}{clcccccc}
\toprule
\multirow{2}[4]{*}{} & \multirow{2}[4]{*}{Index} & \multicolumn{1}{c}{\multirow{2}[4]{*}{\makecell{\hspace{-0.9cm}The period of\\the Historical data}}} & \multicolumn{5}{c}{$\tilde T^\A_n$ [p-value]} \\
\cmidrule{4-8}      &       &       & 0     & 0.1   & 0.2   & 0.3   & 0.5 \\
\midrule
(i)   &\hspace{-0.4cm} S\&P500 & 1/3/2000 - 12/31/2001 & 1.59 [0.44] & 1.30 [0.62] & 1.40 [0.55] & 1.49 [0.50] & 1.66 [0.41] \\
(ii)  &\hspace{-0.4cm} Hang Seng & 1/4/1988 - 12/31/1990 & 0.67 [0.97] & 0.57 [0.99] & 0.62 [0.98] & 0.58 [0.99] & 0.79 [0.93] \\
\bottomrule
\end{tabular}%
}
\end{table}

\noindent {\bf (i) The S\&P500 index from Jan 2000 to Dec 2004}

As can be seen in the left part of Figure \ref{SP:2000}, there seems to be no observations that can be regarded as outlier  compared to the log return series in Figure \ref{HS:1988}. We consider the series from Jan 2000 to Dec 2001 ($n$=499) as the historical data and start monitoring at the first trading day of Jan 2002 ($k$=1). In order to check the parameter constancy in the historical data, we first implement $\tilde T^\A_n$ for each $\A$ considered. The results are summarized in the first row of  Table \ref{Tab:alpha}, where  we can see that all the p-values of $\tilde T^\A_n$ are greater than 10\% and thus accept the null hypothesis that there exists no parameter change in the historical data.
\begin{figure}[!t]
\includegraphics[height=0.4\textwidth,width=\textwidth]{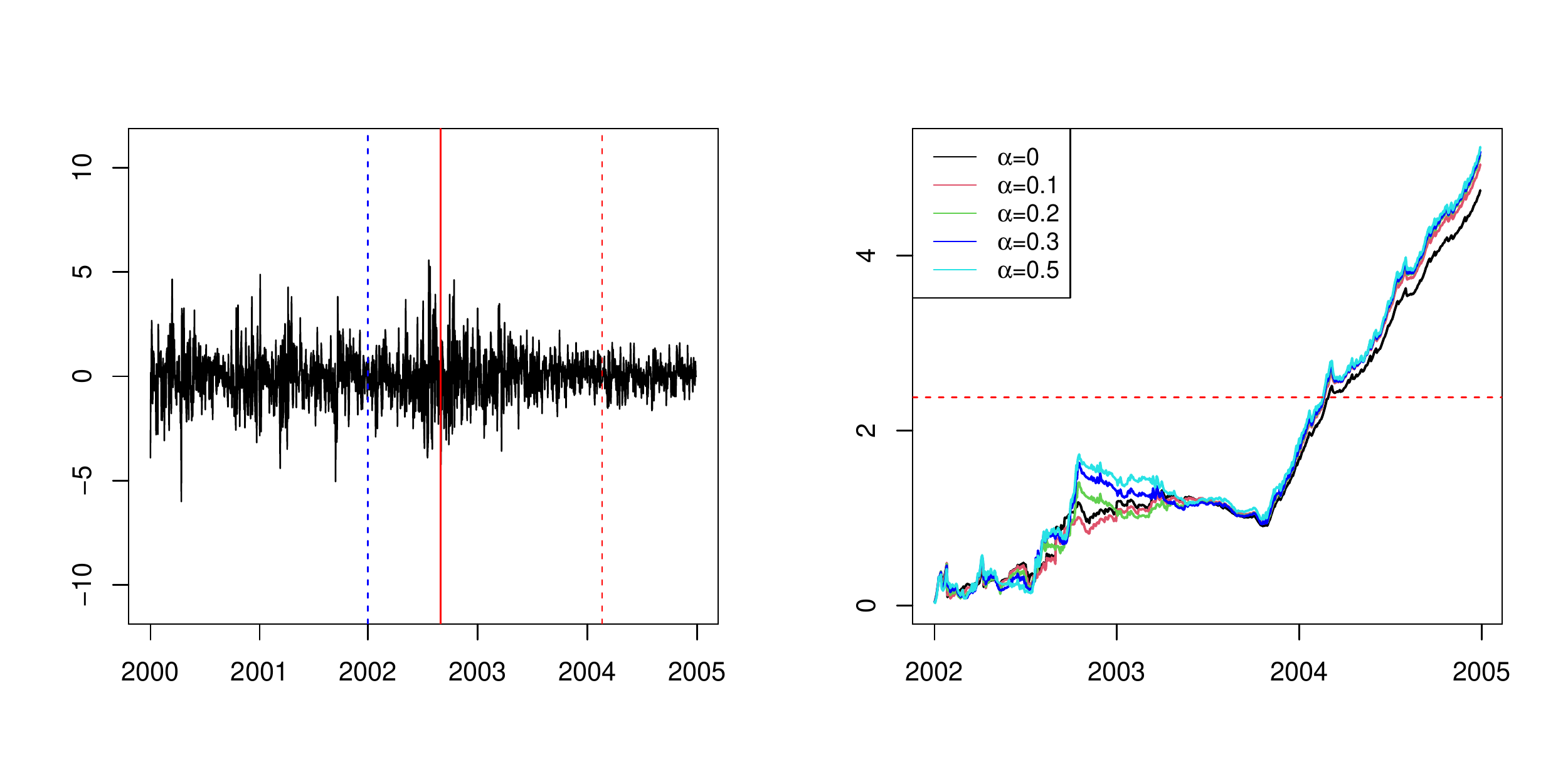}\vspace{-1cm}
\caption{\small The log return series (L) of the S\&P500 index from Jan 3, 2000 to Dec 31, 2004 and the plot (R) of  $\tilde D_{\A,n}(k)$ starting from Jan 2002.  }\label{SP:2000}
\end{figure}

\begin{table}[h]
  \centering
  {\small
  \tabcolsep=6pt
  \renewcommand{\arraystretch}{1.2}
  \caption{\small Results of the change point analyses and parameter estimates for the S\&P500 index.}
   \begin{tabular}{cccccrrrcrrr}
\toprule
\multicolumn{4}{c}{Change point analyses} &       & \multicolumn{7}{c}{Parameter estimates} \\
\cmidrule{1-4}\cmidrule{6-12}      &       & \multicolumn{1}{c}{\multirow{2}[4]{*}{\makecell{$\tilde T^\A_n$\\ \mbox{[p-value]}}}} & \multicolumn{1}{c}{\multirow{2}[4]{*}{\makecell{Estimated\\chg.pt ($c_\A$)}}} &       & \multicolumn{3}{c}{1st period } &       & \multicolumn{3}{c}{2nd period} \\
\cmidrule{6-8}\cmidrule{10-12}\multicolumn{1}{c}{$\A$} & \multicolumn{1}{l}{$n+\tilde k_{\A,n}$} &       &       &       & \multicolumn{1}{c}{$\hat\W$} & \multicolumn{1}{c}{$\hat\A_1$} & \multicolumn{1}{c}{$\hat\B_1$} &       & \multicolumn{1}{c}{$\hat\W$} & \multicolumn{1}{c}{$\hat\A_1$} & \multicolumn{1}{c}{$\hat\B_1$} \\
\cmidrule{1-4}\cmidrule{6-8}\cmidrule{10-12}
0     & 1045   & 4.14 [0.008] & 667   &       & 0.163 & 0.141 & 0.779 &       & 0.012 & 0.051 & 0.930 \\
0.1   & 1039   & 3.81 [0.014] & 667   &       & 0.134 & 0.123 & 0.805 &       & 0.013 & 0.045 & 0.935 \\
0.2   & 1038   & 3.51 [0.024] & 667   &       & 0.120 & 0.113 & 0.817 &       & 0.014 & 0.039 & 0.940 \\
0.3   & 1038   & 3.28 [0.034] & 714   &       & 0.104 & 0.117 & 0.825 &       & 0.006 & 0.001 & 0.985 \\
0.5   & 1037   & 3.04 [0.051] & 714   &       & 0.101 & 0.114 & 0.826 &       & 0.006 & 0.001 & 0.985 \\
\bottomrule
\multicolumn{12}{l}{$\tilde T_{\A,n}$ and the estimated change points  are obtained using the data up to $t=n(=499)+\tilde k_{\A,n}$. }\\
\multicolumn{12}{l}{ The 1st (resp. 2nd) period covers from $t=1$ (resp. $t=c_\A+1$)  to $t=c_\A$ (resp. $t=n+\tilde k_{\A,n}$).}
\end{tabular}}
  \label{Tab:SP500}%
\end{table}

The results of the sequential procedure, that is, the plots of the detector $\tilde D_{\A,n}(k)$, are presented in the right part
of Figure \ref{SP:2000}. One can see that the all the plots show similar shapes going over the critical value about two years later. More exactly, each procedure stops at $\tilde k_{\A,n}$=546, 540, 539, 539, and 538  for $\A$=0, 0.1, 0.2, 0.3, and 0.5, respectively, indicating that the parameter changed somewhere in the monitoring period. To locate the change point, for each $\A$, we apply $\tilde T_{\A,n}$ again to the return series up to each stopped point, that is, $n(=499)+\tilde k_{\A,n}$. The change points, say $c_\A$,  are obtained to be  Aug 30, 2002 ($t$= 667) and Nov 6, 2002 ($t$= 714) by $\A\in \{0, 0.1, 0.2\}$ and $\A\in \{0.3, 0.5\}$, respectively.

The results of change point analyses and  the parameter estimates are summarized in Table \ref{Tab:SP500}, where  the estimates
in the first and  second periods are obtained based on the data up to $t=c_\A$ and the data from $t=c_\A+1$ to $t=n+\tilde k_{\A,n}$, respectively. The estimates in each row are obtained by the MDPDE  with the corresponding $\A$. We recall that the MDPDE with $\A=0$ becomes the MLE, hence the estimates in the first row are the results by the MLE. From the table, we can  clearly  see that the parameters are estimated differently before and after each change point: $\W$ and $\A_1$ are estimated to be smaller and $\B_1$ is obtained to be greater in the 2nd period.  It should be noted that the DPD based inferences prefer to use  $\A$ close to zero when the data is not contaminated or the degree of contamination is small (cf. \cite{BHHJ:1998} and \cite{song:kang:2019}).  Since this example is presumed not to include outlying observations, we rely on the results obtained  by $\A \leq 0.2$. Therefore, we conclude that  $t=667$ is a proper change point. The change point and the earliest stopped date, which is obtained by $\A=0.2$, are additionally depicted by the red solid line and dotted line in the return series  in Figure \ref{SP:2000}, respectively. \vspace{0.3cm}

\begin{figure}[t]
\includegraphics[height=0.4\textwidth,width=\textwidth]{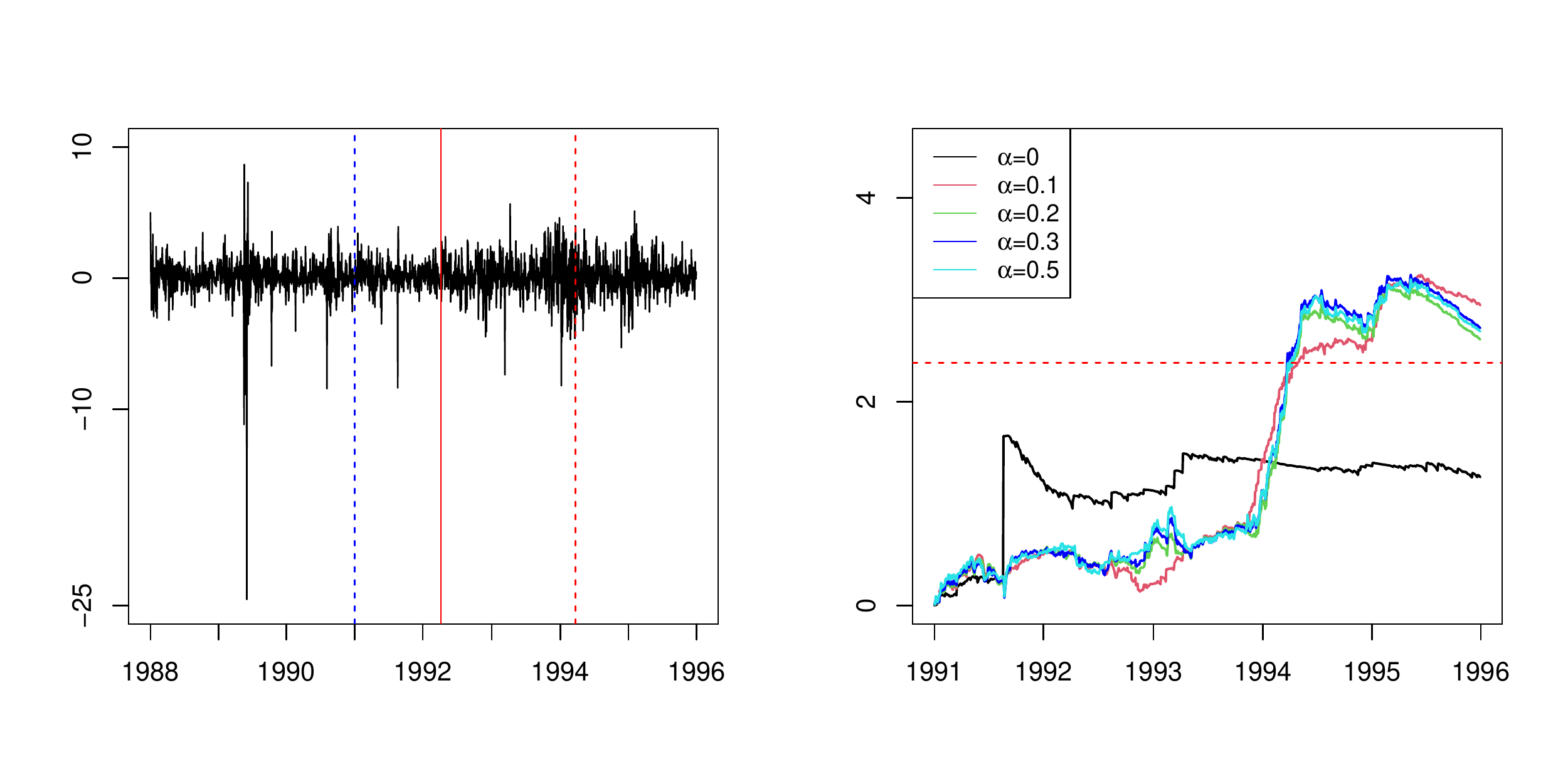}\vspace{-1cm}
\caption{\small The log return series (L) of the Hang Seng index from Jan 4, 1988 to Dec 29, 1995 and the plot (R) of $\tilde D_{\A,n}(k)$ starting from 1991.}\label{HS:1988}
\end{figure}

\noindent {\bf (ii) The Hang Seng index from Jan 1988 to Dec 1996}

In this application, we start the monitoring at January 4, 1991 ($k$=1) and thus the series from 1988 to 1990 is considered as the historical data ($n$=741).  As is clearly shown in the log return series in  Figure \ref{HS:1988}, this data has some apparent outlying observations in the historical data.
Since the score test for parameter change, i.e.,  $\tilde T^\A_n$ with $\A=0$, could be unduly influenced by such observations (cf. \cite{song:kang:2019}), we use  $\tilde T^\A_n$ with $\A>0$ to judge whether or not the parameters change over the historical observations. The test results are presented in the second row of  Table \ref{Tab:alpha}. From the p-values of $\tilde T^\A_n$ for  $\A>0$,  we accept the null hypothesis that the historical data does not undergo parameter change.

The plots of $\tilde D_{\A,n}(k)$ are provided in the right part of Figure \ref{HS:1988}. Unlike the previous analysis, the proposed and the score based procedures yield different results. That is, all $\tilde D_{\A,n}(k)$ with $\A>0$ cross over the dashed horizontal line but the path of $\tilde D_{\A,n}(k)$  with $\A=0$ evolves under the line. In particular, it is noteworthy that  $\tilde D_{\A,n}(k)$  with $\A=0$ shows a jump before long, which may be due to the  deviating observation occurred before 1992,  while all $\tilde D_{\A,n}(k)$ with $\A>0$ move stably  during that period.  This would be taken as an indirect indication that $\tilde D_{\A,n}(k)$ with $\A>0$ is robust against such outlying observations.
Each monitoring stops at $\tilde k_{\A,n}$=828, 804, 803, and 809 for $\A$=0.1, 0.2, 0.3, and 0.5, respectively. After which, to locate a change point that may have occurred before each stopped date, we again conduct the test  $\tilde T^\A_n$ using the data up to $t=n+\tilde k_{\A,n}$.  In the case of $\A=0$, since the score based procedure does not give an alarm signal for parameter change, $\tilde T^{\A=0}_n$ is applied to the whole series.

Table \ref{Tab:HS} shows the test results and the parameter estimates before and after the change points.  We first note that $\tilde T^\A_n$ with $\A=0$ produces the p-value over 10\% and thus it retains the null hypothesis that the parameters do not change over the whole observations. In contrast, all the  p-values of  $\tilde T^\A_n$ with $\A>0$ are obtained to be less than 1\%, rejecting the null hypothesis, and each $\tilde T^\A_n$ locates a change point before each stopped date. For each $\A>0$, we can see the distinct differences between the estimates in the first and second periods. Here, one may doubt that the differences might be due to outliers observed in the first period. We however note that the MDPDE as a robust estimator can reduce the impact of outliers on the estimation. This means that the differences in the estimates are highly likely due to the genuine change in the parameters. Hence, we primarily conclude that the parameter has changed in the monitoring period and presumably conclude that the score based procedure fails to catch the change due to the outlying observations.
 Recalling one of our findings in the simulation study that the procedure with $\A=0.1$ is somewhat affected by outliers, we estimate the change point based on the results obtained by $\A$=0.2 and 0.3. Therefore, we decide that Apr 4, 1992 ($t=1056$) is a suitable change point, which and the stopped point, Mar 23, 1994 ($t=1544$), are also depicted by the red solid and dotted lines in the left part of Figure \ref{HS:1988}, respectively.

\begin{table}[t]
  \centering
  {\small
  \tabcolsep=6pt
  \renewcommand{\arraystretch}{1.2}
  \caption{\small Results of the change point analyses and parameter estimates for the Hang Seng index.}
   \begin{tabular}{cccccrrrcrrr}
\toprule
\multicolumn{4}{c}{Change point analyses} &       & \multicolumn{7}{c}{Parameter estimates} \\
\cmidrule{1-4}\cmidrule{6-12}      &       & \multicolumn{1}{c}{\multirow{2}[4]{*}{\makecell{$\tilde T^\A_n$\\ \mbox{[p-value]}}}} & \multicolumn{1}{c}{\multirow{2}[4]{*}{\makecell{Estimated\\chg.pt ($c_\A$)}}} &       & \multicolumn{3}{c}{1st period} &       & \multicolumn{3}{c}{2nd period} \\
\cmidrule{6-8}\cmidrule{10-12}\multicolumn{1}{c}{$\A$} & $n+\tilde k_{\A,n}$ &       &       &       & \multicolumn{1}{c}{$\hat\W$} & \multicolumn{1}{c}{$\hat\A_1$} & \multicolumn{1}{c}{$\hat\B_1$} &       & \multicolumn{1}{c}{$\hat\W$} & \multicolumn{1}{c}{$\hat\A_1$} & \multicolumn{1}{c}{$\hat\B_1$} \\
\cmidrule{1-4}\cmidrule{6-8}\cmidrule{10-12}
0     & $\cdot$ & 2.34[0.15] & $\cdot$ &       & 0.112 & 0.119 & 0.828 &       &   &  &  \\
0.1   & 1569   & 7.48[0.00] & 1144  &       & 0.258 & 0.149 & 0.584 &       & 0.083 & 0.071 & 0.891 \\
0.2   & 1545   & 6.49[0.00] & 1056  &       & 0.200 & 0.103 & 0.658 &       & 0.044 & 0.057 & 0.919 \\
0.3   & 1544   & 5.79[0.00] & 1056  &       & 0.188 & 0.094 & 0.669 &       & 0.050 & 0.060 & 0.909 \\
0.5   & 1550   & 4.96[0.00] & 1061  &       & 0.178 & 0.082 & 0.679 &       & 0.072 & 0.077 & 0.873 \\
\bottomrule
\multicolumn{12}{l}{In the case of $\A=0$, $\tilde T_{\A,n}$ and the estimates  are obtained based on the whole series.}\\
\multicolumn{12}{l}{ The 1st (resp. 2nd) period covers from $t=1$ (resp. $t=c_\A+1$)  to $t=c_\A$ (resp. $t=n+\tilde k_{\A,n}$).}
\end{tabular}}
  \label{Tab:HS}%
\end{table}

\section{Concluding remark}\label{Sec:con}
This study proposed a robust sequential procedure for monitoring parameter changes. We constructed the DP divergence based detector and investigated the asymptotic behaviors of the induced stopping time under the null and alternative hypotheses. In particular, we provided a set of sufficient conditions for time series models under which the proposed procedure has an asymptotically controlled size and consistency in powers. The simulation study showed that the score based procedure is sensitively influenced by outliers whereas our procedure is strongly robust to outliers. The usefulness of our procedure was also demonstrated in real data analysis, where the proposed procedure detected a change point that was missed by the score based procedure.

Extension to other models including multivariate models are also interesting. In particular, we expect that our sequential procedure can naturally applied to inter-valued time series models that have recently attracted lots of attention. We leave the extension to multivariate models as a possible topic of future study.


\section{Appendix}
In this appendix, we provide the proofs of Theorems \ref{thm2}-\ref{thm5} for the case of $\A > 0$. We begin with two technical lemmas, which are usefully used in proving the lemmas and theorems below.

\begin{lm}\label{lm1}
Let $\{X_n| n\geq1\}$ be a sequence of random variables on a probability space $(\Omega,\mathcal{F},P)$ such that $X_n$ converges almost surely to a random variable $X$. Then, for any fixed integer $m$,
\[ \lim_{n\rightarrow\infty}\sup_{k\geq m} X_{n+k} =X \quad a.s.\]
\end{lm}
\begin{proof}

Let $S:= \{ \omega \in \Omega | X_n(\omega)\rightarrow X(\omega)\}$ and $Z_n:=\sup_{k\geq m} |X_{n+k}-X|$. Then,
 for a fixed $\omega \in S$ and any $\ep>0$, there exists an integer $N$ such that $|X_n(\omega)-X(\omega)|\leq \ep$ for all $n\geq N+m$. This implies that $Z_N(\omega)=\sup_{k\geq m} |X_{N+k}(\omega)-X(\omega)| \leq\epsilon$. Since $\{Z_n(\omega)\}$ is a decreasing sequence, one can see that $Z_n(\omega) \leq \ep$ for all $n\geq N$. Thus, we have that for all $n\geq N$,
$$\Big|\sup_{k\geq m}X_{n+k}(\omega)-X(\omega)\Big| \leq \sup_{k\geq m}|X_{n+k}(\omega)-X(\omega)| \leq \ep,$$
which asserts the lemma.
\end{proof}

\begin{lm}\label{unif}
Let $C(\Theta,\mathbb{R}^{d'})$ be a space of continuous functions from $\Theta$ to $\mathbb{R}^{d'}$, where $\Theta$ is a compact subset of $\mathbb{R}^d$, and $\{f_t| t\geq1\}$ a stationary ergodic sequence of random elements in $C(\Theta,\mathbb{R}^{d'})$.
If $\E \sup_{\T\in\Theta} \|f_t(\T)\| <\infty$, then for a constant $c>0$,
\[ \sup_{\T\in\Theta} \Big\| \frac{1}{k_n-c}\sum_{t>n+c}^{n+k_n} f_t(\T) - \E f_1(\T) \Big\|=o(1)\quad a.s., \]
where $\{k_n\}$ is an increasing sequence of positive integers.
\end{lm}
\begin{proof}
It is sufficient to show the above for $d'=1$, so we assume that $f_t(\T)$ is a random element in $C(\Theta,\mathbb{R})$. We follow the arguments in the proof of Theorem 16(a) in \cite{ferguson:1996}.

Letting $\varphi_t(\T,r):=\sup \{f_t(\T')\, |\, \T' \in N(\T,r)\}$, where $N(\T,r)=\{\T'\, |\, \|\T'-\T\|<r\}$, it follows from the monotone convergence theorem that
\[ \lim_{r\rightarrow0+}\E\, \varphi_t(\T,r) = \E\,f_t (\T):=f(\T).\]
For $\ep>0$ and $\T$, let $r_{\T,\ep}$, say $r_\T$ for short, be a positive constant satisfying $\E\, \varphi_t(\T,r_\T)\leq f(\T)+\ep$. Then, due to the compactness of $\Theta$, we can take a finite subcover $\{N(\T_j,r_{\T_j})\}_{j=1}^m$, that is, $\Theta \subset \cup_{j=1}^m N(\T_j,r_{\T_j})$, and thus one can see that
for all $\T\in\Theta$,
\begin{eqnarray*}
\frac{1}{k_n-c}\sum_{t>n+c}^{n+k_n} f_t(\T) \leq  \max_{1\leq j \leq m}\frac{1}{k_n-c}\sum_{t>n+c}^{n+k_n} \varphi_t(\T_j,r_{\T_j}).
\end{eqnarray*}
Applying the ergodic theorem to  $\frac{1}{k_n-c}\sum_{t>n+c}^{n+k_n} \varphi_t(\T_j,r_{\T_j})$, we have that almost surely,
\begin{eqnarray*}
\limsup_{n\rightarrow\infty} \sup_{\T\in\Theta}\frac{1}{k_n-c}\sum_{t>n+c}^{n+k_n} f_t(\T)
\leq \max_{1\leq j \leq m} \E\,\varphi_t(\T_j,r_{\T_j}) \leq  \sup_{\T\in\Theta} f(\T)+\ep.
\end{eqnarray*}
and thus letting $\ep\rightarrow0$, we obtain
\begin{eqnarray}\label{uni1}
\limsup_{n\rightarrow\infty} \sup_{\T\in\Theta}\frac{1}{k_n-c}\sum_{t>n+c}^{n+k_n} f_t(\T)
&\leq& \sup_{\T\in\Theta} f(\T)\quad a.s.
\end{eqnarray}
Using (\ref{uni1}) and  exactly the same arguments in page 110 in \cite{ferguson:1996}, one can show the uniform strong convergence for $d'=1$. This completes the proof.
\end{proof}
\subsection{Lemmas and proofs for Section \ref{Sec:2}}
\begin{lm}\label{lm2}Suppose that assumptions {\bf A1}-{\bf A7} hold and let $m$ be a nonnegative integer. Then, under  $H_0$,
\begin{eqnarray*}
\sup_{ k\geq m} \Big\|\frac{1}{n+k}\paa H_{\A,n+k}(\theta^*_{n,k})-\mathcal{J}_\A \Big\|=o(1)\quad a.s.,
\end{eqnarray*}
where $\{\theta^*_{n,k}\}$ is any double array of random vectors satisfying that $\|\theta^*_{n,k}-\theta_0\|\leq\|\hat\theta_{\A,n}-\theta_0\|$.
\end{lm}
\begin{proof}
First, observe that by assumption {\bf A6},
\begin{eqnarray}\label{mom.cond1}
\E \sup_{\T \in N(\T_0)} \big\|\paa l_\A(X;\T)-\paa l_\A(X;\T_0)\big\| <\infty.
\end{eqnarray}
Let $N_r(\theta_0) := \{\theta \in \Theta :\|\theta-\theta_0\|\leq r\}$. Noting that $\paa l_\A(X;\T)$ is continuous in $\T$, one can see from the dominate convergence theorem that
\[ \lim_{n\rightarrow\infty}\E \sup_{\T \in N_{1/n}(\T_0)} \big\|\paa l_\A(X;\T)-\paa l_\A(X;\T_0)\big\|=0.\]
Thus, for any $\ep>0$, we can take a neighborhood $N_\ep(\T_0)$ such that
\begin{eqnarray}\label{mom.cond2}
\E \sup_{\T \in N_\ep(\T_0)} \big\|\paa l_\A(X;\T)-\paa l_\A(X;\T_0)\big\| <\ep.
\end{eqnarray}
Since $\hat{\theta}_{\A,n}$ converges almost surely to $\theta_0$, we have that for sufficiently large $n$,
\begin{eqnarray*}
&&\hspace{-0.5cm}\sup_{k\geq m} \Big\|\frac{1}{n+k}\paa H_{\A,n+k}(\theta^*_{n,k})-\mathcal{J}_\A\Big\|\\
&&\leq  \sup_{k\geq m} \frac{1}{n+k}\Big\|\paa  H_{\A,n+k}(\theta^*_{n,k})-\paa H_{\A,n+k}(\theta_0)\Big\|
+ \sup_{k\geq m} \Big\|\frac{1}{n+k}\paa H_{\A,n+k}(\theta_0)-\mathcal{J}_\A\Big\|\\
&&\hspace{-0cm}\leq
\sup_{k\geq m}\frac{1}{n+k} \sum_{t=1}^{n+k} \sup_{\theta\in N_{\ep}(\theta_0)} \big\| \paa l_\A (X_t;\theta)-\paa l _\A (X_t;\theta_0)\big\|
+ \sup_{k\geq m} \Big\|\frac{1}{n+k}\paa  H_{\A,n+k}(\theta_0)-\mathcal{J}_\A\Big\|\\
&&\hspace{-0.5cm}:=I_n +II_n.
\end{eqnarray*}
Noting that $\{l_\A (X_t;\theta)\}$ is a sequence of  i.i.d. random variables, we have that
\begin{eqnarray*}
&&I_n^o:=\frac{1}{n} \sum_{t=1}^{n} \sup_{\theta\in N_{\ep}(\theta_0)} \big\| \paa l_\A (X_t;\theta)-\paa l_\A (X_t;\theta_0)\big\|\
\stackrel{a.s.}{\longrightarrow}\ \E\sup_{\theta \in N_\ep(\theta_0)}\big\| \paa l_\A (X_t;\theta)- \paa l_\A (X_t;\theta_0)\big\|, \\
&& II_n^o:=\Big\|\frac{1}{n}\paa H_{\A,n}(\theta_0)-\mathcal{J}_\A\Big\|\ \stackrel{a.s.}{\longrightarrow}\ 0.
\end{eqnarray*}
Hence, it follows from Lemma \ref{lm1} that
\begin{eqnarray*}
&& I_n =\sup_{k\geq m} I_{n+k}^o\stackrel{a.s.}{\longrightarrow}\ \E\sup_{\theta \in N_\ep(\theta_0)}\big\| \paa l_\A (X_t;\theta)- \paa l_\A (X_t;\theta_0)\big\| < \ep,\\
&&II_n=\sup_{k\geq m} II_{n+k}^o\stackrel{a.s.}{\longrightarrow}\ 0,
\end{eqnarray*}
which establish the lemma.
\end{proof}

\begin{lm}\label{lm3}Suppose that assumptions {\bf A1}-{\bf A7} hold. Then, under  $H_0$,
\begin{eqnarray*}
\sup_{k\geq1} \frac{1}{n+k} \big\|\paa H_{\A,n+k}(\theta^*_{n,k}) \mathcal{J}_\A^{-1}(B_{\A,n}-\mathcal{J}_\A)\big\|=o(1)\quad a.s.,
\end{eqnarray*}
where $\{\theta^*_{n,k}\}$ is the one given in Lemma \ref{lm2} and  $B_{\A,n}=\paa H_{\A,n}(\theta^*_{n,0})/n$.
\end{lm}
\begin{proof}
By Lemma \ref{lm2}, we have
\[\sup_{k\geq1} \frac{1}{n+k} \big\|\paa  H_{\A,n+k}(\theta^*_{n,k}) \big\|
\leq\sup_{k\geq1}  \Big\|\frac{1}{n+k}\paa  H_{\A,n+k}(\theta^*_{n,k})-\mathcal{J}_\A\Big\|
+\|\mathcal{J}_\A\|=O(1)\quad a.s.\]
and
\begin{eqnarray}\label{BJ}
\big\|\mathcal{J}_\A^{-1}(B_{\A,n}-\mathcal{J}_\A)\big\| \leq \big\|\mathcal{J}_\A^{-1}\big\|\,\sup_{k\geq0}  \Big\|\frac{1}{n+k}\paa H_{\A,n+k}(\theta^*_{n,k})-\mathcal{J}_\A\Big\|=o(1)\quad a.s.,
\end{eqnarray}
which yield the result.
\end{proof}

\begin{lm}\label{lm4}Suppose that assumptions {\bf A1}-{\bf A7} hold. Then, under  $H_0$,
\begin{eqnarray*}
\sup_{k\geq1}\frac{1}{\sqrt{n}\Big(1+\frac{k}{n}\Big)}\Big\| \pa_{\theta}H_{\A,n+k}(\hat{\theta}_{\A,n})-\pa_{\theta}H_{\A,n+k}(\theta_0)+\Big(1+\frac{k}{n}\Big)\pa_{\theta}H_{\A,n}(\theta_0)\Big\|=o_P(1).
\end{eqnarray*}
\end{lm}
\begin{proof}
 By Taylor's theorem, we have  that  for each $ k\geq0$,
\begin{eqnarray}\label{TL}
\pa_{\theta}H_{\A,n+k}(\hat{\theta}_{\A,n})=\pa_{\theta}H_{\A,n+k}(\theta_0)+\paa H_{\A,n+k}(\theta^*_{n,k}) (\hat{\theta}_{\A,n}-\theta_0),
\end{eqnarray}
where $\theta^*_{n,k}$ is an intermediate point between $\theta_0$ and $\hat{\theta}_{\A,n}$. Here, observe that for $k=0$,
\[
\pa_{\theta}H_{\A,n}(\hat\theta_{\A,n})=\pa_{\theta}H_{\A,n}(\theta_0)+\paa H_{\A,n}(\theta^*_{n,0}) (\hat{\theta}_{\A,n}-\theta_0)=0.
\]
Then,  we can express that
\begin{eqnarray}\label{hat.theta}
\hat{\theta}_{\A,n}-\theta_0
= -\mathcal{J}_\A^{-1}\frac{1}{n}\pa_{\theta}H_{\A,n}(\theta_0)-\mathcal{J}_\A^{-1}(B_{\A,n}-\mathcal{J}_\A)(\hat{\theta}_{\A,n}-\theta_0),
\end{eqnarray}
where $B_{\A,n}=\paa H_{\A,n}(\theta^*_{n,0})/n$. Putting the above into (\ref{TL}), we obtain
\[\pa_{\theta}H_{\A,n+k}(\hat{\theta}_{\A,n})-\pa_{\theta}H_{\A,n+k}(\theta_0)=-\paa H_{\A,n+k}(\theta^*_{n,k})\Big\{
\mathcal{J}_\A^{-1}\frac{1}{n}\pa_{\theta}H_{\A,n}(\theta_0)+\mathcal{J}_\A^{-1}(B_{\A,n}-\mathcal{J}_\A)(\hat{\theta}_{\A,n}-\theta_0)\Big\}\]
and thus
\begin{eqnarray*}
&&\hspace{-0.5cm}\frac{1}{\sqrt{n} \Big(1+\frac{k}{n}\Big)}\Big\|\pa_{\theta}H_{\A,n+k}(\hat{\theta}_{\A,n})-\pa_{\theta}H_{\A,n+k}(\theta_0)+\Big(1+\frac{k}{n}\Big)\pa_{\theta}H_{\A,n}(\theta_0)\Big\|\\
&&\hspace{-0cm}\leq
\Big\|\Big( \mathcal{J}_\A-\frac{1}{n+k}\paa H_{\A,n+k}(\theta^*_{n,k})\Big)\,\mathcal{J}_\A^{-1}\frac{1}{\sqrt{n}}\pa_{\theta}H_{\A,n}(\theta_0)\Big\|\\
&&\hspace{0.5cm}
+\frac{1}{n+k}\big\|\paa H_{\A,n+k}(\theta^*_{n,k})\,\mathcal{J}_\A^{-1}
(B_{\A,n}-\mathcal{J}_\A)\sqrt{n}(\hat{\theta}_{\A,n}-\theta_0)\big\|\\
&&\hspace{-0.5cm}:=I_{n,k} +II_{n,k}.
\end{eqnarray*}
First note that $\E[\pa_\T l_\A(X;\T_0)]=0$ and that $\{\pa_\T l_\A(X_t;\T_0)\}$ is a sequence of i.i.d. random vectors. Then, by the multivariate functional central limit theorem (FCLT), we have that for each $T>0$,
\begin{eqnarray}\label{FCLT}
\frac{1}{\sqrt{n}} \pa_\T H_{\A,[ns]}(\T_0)=\frac{1}{\sqrt{n}}  \sum_{t=1}^{[ns]}  \pa_\T l_\A (X_t; \T_0)\stackrel {w}{\longrightarrow}\mathcal{I}_{\alpha}^{1/2}W_d(s)~~in~~\mathbb {D}([0,T],\mathbb {R}^d ),
\end{eqnarray}
where $\{W_d(s)\}$ is a  $d$-dimensional standard Wiener process. Hence, it follows that $\frac{1}{\sqrt{n}} \pa_\T H_{\A,n}(\T_0)$ is $O_P(1)$, which together with Lemma \ref{lm2} yields  $\sup_{k\geq1}I_{n,k}=o_P(1)$.\\
Next, to show  $\sup_{k\geq1}II_{n,k}=o_P(1)$, we prove that  $\sqrt n(\hat\T_{\A,n}-\T_0)=O_P(1)$. From (\ref{hat.theta}), we have
\begin{eqnarray}\label{root.hat}
 \sqrt{n} (\hat{\theta}_{\A,n}-\theta_0)
= -\mathcal{J}_\A^{-1}\frac{1}{\sqrt{n}}\pa_{\theta}H_{\A,n}(\theta_0)-
\mathcal{J}_\A^{-1}(B_{\A,n}-\mathcal{J}_\A)\sqrt{n}(\hat{\theta}_{\A,n}-\theta_0).
\end{eqnarray}
Since $\mathcal{J}_\A^{-1}(B_{\A,n} -\mathcal{J}_\A)=o_P(1)$ by (\ref{BJ}) and
 $\frac{1}{\sqrt{n}} \pa_\T H_{\A,n}(\T_0)=O_P(1)$, it can be shown that $\sqrt n(\hat\T_{\A,n}-\T_0)=O_P(1)$ (cf. Lemma \ref{root.n.con} below). Thus, due to Lemma \ref{lm3}, we have $\sup_{k\geq1}II_{n,k}=o_P(1)$. This completes the proof.
\end{proof}

\begin{lm}\label{lm5}Suppose that assumptions {\bf A1}-{\bf A5}, {\bf A7}, and {\bf B} hold. Then, under  $H_0$,
\begin{eqnarray*}
\sup_{k\geq1}  \frac{\Big\|\pa_{\theta}{H}_{\A,n+k}(\theta_0)-\Big(1+\frac{k}{n}\Big)\pa_{\theta}{H}_{\A,n}(\theta_0)\Big\|}{\sqrt{n}\Big(1+\frac{k}{n}\Big)b\Big(\frac{k}{n}\Big)}
\stackrel{d}{\longrightarrow} \sup_{s>0} \frac{ \big\|\mathcal{I}_\A^{\frac{1}{2}}\big( W_d(1+s)-(1+s)W_d(1)\big)\big\|}{(1+s)b(s)}.
\end{eqnarray*}
\end{lm}
\begin{proof}
We follow the arguments in Lemma 6.6 in \cite{berkes.et.al:2004seq} to show the lemma.
Using (\ref{FCLT}), we have that for any $T>0$,
\[\frac{1}{\sqrt{n}}\Big\{\pa_{\theta}{H}_{\A,[n(1+s)]}(\theta_0)-\Big(1+\frac{[ns]}{n}\Big)\pa_{\theta}{H}_{\A,n}(\theta_0)\Big\}  \stackrel {w}{\longrightarrow}\mathcal{I}_\A^{\frac{1}{2}}\big( W_d(1+s)-(1+s)W_d(1)\big)~~in~~\mathbb {D}([0,T],\mathbb {R}^d )\]
and thus, by the continuous mapping theorem,
\begin{eqnarray}\label{R1}
\sup_{1\leq k \leq nT}  \frac{\big\|\pa_{\theta}{H}_{\A,n+k}(\theta_0)-\Big(1+\frac{k}{n}\Big)\pa_{\theta}{H}_{\A,n}(\theta_0)\big\|}{\sqrt{n}\Big(1+\frac{k}{n}\Big)b\Big(\frac{k}{n}\Big)}
\stackrel{d}{\longrightarrow} \sup_{0<s<T} \frac{ \big\|\mathcal{I}_\A^{\frac{1}{2}}\big( W_d(1+s)-(1+s)W_d(1)\big)\big\|}{(1+s)b(s)}.
\end{eqnarray}
Further, since
\begin{eqnarray*}
\lim_{T\rightarrow \infty} \sup_{s\geq T} \frac{\big\| W_d (1+s)\big\|}{(1+s)b(s)} =0\quad a.s.
\end{eqnarray*}
by the law of the iterated logarithm and assumption {\bf B}, we have
\begin{eqnarray}\label{R2}
\sup_{0<s<T} \frac{ \big\| W_d(1+s)-(1+s)W_d(1)\big\|}{(1+s)b(s)}
\stackrel{a.s.}{\longrightarrow}\sup_{s>0} \frac{  \big\|W_d(1+s)-(1+s)W_d(1)\big\|}{(1+s)b(s)}\quad as\quad T\rightarrow \infty.
\end{eqnarray}
For  any $\ep>0$, note also  that by  the Hájek-Rényi-Chow inequality,
\begin{eqnarray}\label{R3}
\lim_{T\rightarrow\infty} \lim_{n\rightarrow\infty} P\bigg( \sup_{ k \geq nT} \frac{1}{\sqrt{n}\Big(1+\frac{k}{n}\Big)}\big\|\pa_{\theta}{H}_{\A,n+k}(\theta_0)\big\| \geq \epsilon\bigg)=0.
\end{eqnarray}
Then, the lemma is established by combining (\ref{R1})-(\ref{R3}).
\end{proof}

\newpage
\noindent{\bf Proof of Theorem \ref{thm2}}\\
Using  the results of Lemmas \ref{lm4} and  \ref{lm5}, we have that  under $H_0$,
\begin{eqnarray}\label{main1}
\sup_{k\geq1}  \frac{\big\|\hat{\mathcal{I}}_{\A,n}^{-\frac{1}{2}}\,\pa_{\theta}H_{\A,n+k}(\hat{\theta}_{\A,n})\big\|}{\sqrt{n}\Big(1+\frac{k}{n}\Big)b\Big(\frac{k}{n}\Big)}
\stackrel{d}{\longrightarrow} \sup_{s>0} \frac{ \big\| W_d(1+s)-(1+s)W_d(1)\big\|}{(1+s)b(s)}.
\end{eqnarray}
Since
the two processes $\{ W_d(1+s)-(1+s)W_d(1)\,  |\, s>0\}$ and  $\{ (1+s)W_d(s/(1+s))\, |\, s>0\}$ have the same distribution, it follows from (\ref{main1}) that
\begin{eqnarray*}
\lim_{n\rightarrow\infty} P\big( k_n <\infty\ |\ H_0 \big)
&=&\lim_{n\rightarrow\infty} P \bigg(  \sup_{k\geq1}  \frac{\big\|\hat{\mathcal{I}}_{\A,n}^{-\frac{1}{2}}\,\pa_{\theta}H_{\A,n+k}(\hat{\theta}_{\A,n})\big\|}{\sqrt{n}\Big(1+\frac{k}{n}\Big)b\Big(\frac{k}{n}\Big)} > 1\ \big|\ H_0\bigg)\\
&=&P\bigg(  \sup_{s>0} \frac{ \big\| W_d(1+s)-(1+s)W_d(1)\big\|}{(1+s)b(s)} >1 \bigg)\\
&=&
P\bigg(  \sup_{s>0} \frac{ \big\| W_d(s/(1+s))\big\|}{b(s)} >1 \bigg)
=
P\bigg(  \sup_{0<s<1} \frac{ \big\| W_d(s)\big\|}{b(s/(1-s))} >1\bigg).
\end{eqnarray*}
\hfill{$\Box$}\vspace{0.2cm}\\

\noindent{\bf Proof of Theorem \ref{thm3}}\\
It suffices to show that there exists a sequence of real numbers, say $\{k_n|n\geq1\}$, such that
\begin{eqnarray}\label{H1}
\frac{\big\| \pa_{\theta}H_{\A,n+k_n}(\hat{\theta}_{\A,n})\big\|}{\sqrt{n}\Big(1+\frac{k_n}{n}\Big)}\stackrel{P}{\longrightarrow} \infty \quad as\ \  n\rightarrow \infty.
\end{eqnarray}
Let $\{k_n\}$ be an increasing sequence of positive integers satisfying $k_n/\sqrt{n}\rightarrow \infty$. Without loss of generality, we assume that $k_n$ is greater than $k^*$. Now, we shall show that (\ref{H1}) holds for the sequence $\{k_n\}$.
Under $H_1$, it can be written that
\begin{eqnarray*}
\pa_{\theta}H_{\A,n+k_n}(\hat{\theta}_{\A,n})
&=&\sum_{t=1}^{n+k^*}\pa_{\theta}\,l_\A(X_{0,t};\hat\T_{\A,n})+\sum_{t>n+k^*}^{n+k_n}\pa_{\theta}\,l_\A(X_{1,t};\hat\T_{\A,n}).
\end{eqnarray*}
Let $H_{\A,n}^0(\T)=\sum_{t=1}^n l_\A(X_{0,t};\T)$. Then, since  $\{X_{0,t} | t\geq 1\}$ is an i.i.d. sample from $f_{\T_0}$, one can see from Lemmas \ref{lm4} and \ref{lm5} that
\begin{eqnarray*}
\frac{1}{\sqrt{n}\Big(1+\frac{k^*}{n}\Big)}\Big\|\sum_{t=1}^{n+k^*}\pa_{\theta}\,l_\A(X_{0,t};\hat\T_{\A,n})\Big\|
\leq \sup_{k\geq1}\frac{1}{\sqrt{n}\Big(1+\frac{k}{n}\Big)}\big\| \pa_{\theta}H^0_{\A,n+k}(\hat{\theta}_{\A,n})\big\|=O_P(1),
\end{eqnarray*}
and thus we first obtain
\begin{eqnarray}\label{AC1}	
\frac{1}{\sqrt{n}\Big(1+\frac{k_n}{n}\Big)}\Big\|\sum_{t=1}^{n+k^*}\pa_{\theta}\,l_\A(X_{0,t};\hat\T_{\A,n})\Big\|=\frac{1}{\sqrt{n}\Big(1+\frac{k_n}{n}\Big)}O_P\Big(\sqrt{n}\Big(1+\frac{k^*}{n}\Big)\Big)=O_P(1).
\end{eqnarray}

Next, noting that $\pa_{\theta}\,l_\A(x; \T) $ is continuous in $\T$, one can see from assumption {\bf A8} and the dominate convergence theorem that $\E\,\pa_{\theta}\,l_\A(x; \T)$ becomes a continuous function on $N'(\T_0)$. By Lemma \ref{unif} with assumption {\bf A8}, we also have
\begin{eqnarray*}\label{USL1}
\sup_{\T\in N'(\T_0)}\Big\| \frac{1}{k_n-k^*}\sum_{t>n+k^*}^{n+k_n}\pa_{\theta}\,l_\A(X_{1,t}; \T) -\E\, \pa_{\theta}\,l_\A(X_{1,t};\T)\Big\|=o(1)\quad a.s.,
\end{eqnarray*}
which  together with the fact that $\hat\T_{\A,n}$ converges almost surely to $\T_0$ yields
\begin{eqnarray*}
\Big\| \frac{1}{k_n-k^*}\sum_{t>n+k^*}^{n+k_n}\pa_{\theta}\,l_\A(X_{1,t}; \hat\T_{\A,n})\Big\|\stackrel{a.s.}{\longrightarrow} \big\|\E \pa_{\theta}\,l_\A(X_{1,t};\T_0)\big\|.
\end{eqnarray*}
Thus, since  $k_n/\sqrt{n}\rightarrow\infty$, we have by assumption {\bf A9} that
\begin{eqnarray}\label{AC2}
\frac{1}{\sqrt{n}\Big(1+\frac{k_n}{n}\Big)}\Big\|\sum_{t>n+k^*}^{n+k_n}\pa_{\theta}\,l_\A(X_{1,t};\hat \T_{\A,n})\Big\|
=\sqrt{n}\,\frac{k_n-k^*}{n+k_n}\big(\big\| E \pa_{\theta}\,l_\A(X_{1,t};\T_0)\big\|+o(1)\big)\stackrel{a.s.}{\longrightarrow}\infty.
\end{eqnarray}
Therefore, it follows from (\ref{AC1}) and (\ref{AC2}) that
\begin{align*}
\sup_{k\geq 1}\frac{\big\| \pa_{\theta}H_{\A,n+k}(\hat{\theta}_{\A,n})\big\|}{\sqrt{n}\Big(1+\frac{k}{n}\Big)}&\ \geq\ \frac{\big\| \pa_{\theta}H_{\A,n+k_n}(\hat{\theta}_{\A,n})\big\|}{\sqrt{n}\Big(1+\frac{k_n}{n}\Big)}\\
&\ \geq\
 \frac{1}{\sqrt{n}\Big(1+\frac{k_n}{n}\Big)}\Bigg\{\Big\|\sum_{t>n+k^*}^{n+k_n}\pa_{\theta}\,l_\A(X_{1,t};\hat \T_{\A,n})\Big\|
-\Big\|\sum_{t=1}^{n+k^*}\pa_{\theta}\,l_\A(X_{0,t};\hat\T_{\A,n})\Big\|\Bigg\}\\
&\ \stackrel{P}{\longrightarrow}\infty
\end{align*}
and consequently we have
\begin{eqnarray*}
\lim_{n\rightarrow\infty} P\big( k_{\A,n} <\infty\ |\ H_1 \big)=1.
\end{eqnarray*}
This completes the proof.
\hfill{$\Box$}\\

\subsection{Lemmas and proofs for Section \ref{Sec:3}}
From now,  $\hat\T_{\A,n}$, $\mathcal{I}_\A$, $\mathcal{J}_\A$, $l_\A(X;\T)$, and $H_{\A,n}(\T)$ are the ones given in Section \ref{Sec:3}.

\begin{lm}\label{lm.I0}Suppose that  assumptions {\bf M1}-{\bf M4} and conditions {\bf S3} and {\bf S4} hold. For a nonnegative integer $m$, we have that  under  $H_0$,
\begin{eqnarray}\label{C1}
\sup_{ k\geq m} \Big\|\frac{1}{n+k}\paa \tilde{H}_{\A,n+k}(\theta^*_{n,k})-\mathcal{J}_\A\Big\|=o(1)\quad a.s.
\end{eqnarray}
and
\begin{eqnarray}\label{C2}
\sup_{k\geq1} \frac{1}{n+k} \big\|\paa \tilde H_{\A,n+k}(\theta^*_{n,k}) \mathcal{J}_\A^{-1}(\tilde B_{\A,n}-\mathcal{J}_\A)\big\|=o(1)\quad a.s.,
\end{eqnarray}
where $\{\theta^*_{n,k}\}$ is any double array of random vectors satisfying that $\|\T^*_{n,k}-\T_0\|\leq\|\hat\T_{\A,n}-\T_0\|$ and $\tilde B_{\A,n}=\paa \tilde H_{\A,n}(\theta^*_{n,0})/n$.
\end{lm}
\begin{proof}
For any $\ep>0$, by assumption {\bf M2} and  condition {\bf S4}, one can take a positive constant  $r_\ep$ such that
\begin{eqnarray}\label{P3.1}
\E\sup_{\theta \in N_\ep(\theta_0)}\big\| \paa l_\A (X_t;\theta)- \paa l_\A (X_t;\theta_0)\big\| < \ep,
\end{eqnarray}
where $N_\epsilon(\theta_0)=\{\theta\in N_1(\T_0)\cap N_2(\T_0)\, |\,\|\theta-\theta_0\|< r_\epsilon\}$.\\
Since we assume in assumption {\bf M4} that $\hat{\theta}_{\A,n}$ converges almost surely to $\theta_0$, we have that for sufficiently large $n$,
\begin{eqnarray*}
&&\hspace{-0.7cm}\sup_{k\geq m} \Big\|\frac{1}{n+k}\paa \tilde{H}_{\A,n+k}(\theta^*_{n,k})-\mathcal{J}_\A \Big\|\\
&&\hspace{-0.7cm}\leq \sup_{k\geq m} \frac{1}{n+k}\big\|\paa \tilde{H}_{\A,n+k}(\theta^*_{n,k})-\paa H_{\A, n+k}(\theta^*_{n,k})\big\|
+ \sup_{k\geq m} \frac{1}{n+k}\big\|\paa H_{\A,n+k}(\theta^*_{n,k})-\paa H_{\A,n+k}(\theta_0)\big\|\\
&&\hspace{-0.3cm}
+ \sup_{k\geq m} \Big\|\frac{1}{n+k}\paa H_{\A,n+k}(\theta_0)-\mathcal{J}_\A\Big\|\\
&&\hspace{-0.7cm}\leq \sup_{k\geq m}\frac{1}{n+k} \sum_{t=1}^{n+k} \sup_{\theta\in N_{\ep}(\theta_0)} \big\| \paa\,\tilde{l}_\A (X_t;\theta)-\paa\,l_\A (X_t;\theta)\big\|\\
&&\hspace{-0.3cm}
+\sup_{k\geq m}\frac{1}{n+k} \sum_{t=1}^{n+k} \sup_{\theta\in N_{\ep}(\theta_0)} \big\| \paa\,l_\A (X_t;\theta)-\paa\,l_\A (X_t;\theta_0)\big\|
+ \sup_{k\geq m} \Big\|\frac{1}{n+k}\sum_{t=1}^{n+k}\paa l_\A(X_t;\theta_0)-\mathcal{J}_\A\Big\|\\
&&\hspace{-0.7cm}:=\sup_{k\geq m}I_{n+k}^o +\sup_{k\geq m}II_{n+k}^o+ \sup_{k\geq m}III_{n+k}^o\quad a.s.
\end{eqnarray*}
Thanks to Lemma \ref{lm1}, it suffices to show that all of  $I_n^o$, $II_n^o$, and $III_n^o$ converge almost surely to zero.
First, we have by condition {\bf S3} that $I_n^o=o(1)$ a.s. One can also show $II_n^o=o(1)$ by using (\ref{P3.1}) and the ergodic theorem.
The last one follows also from the ergodic theorem. Hence, (\ref{C1}) is established.  Following the same argument as in Lemma \ref{lm3}, one can also obtain (\ref{C2}).
\end{proof}

\begin{lm}\label{lm.I1}Suppose that assumption {\bf M1}-{\bf M4} and conditions {\bf S1}-{\bf S4} hold. Then, under  $H_0$,
\begin{eqnarray*}
\sup_{k\geq1}\frac{1}{\sqrt{n}\Big(1+\frac{k}{n}\Big)}\Big\| \pa_{\theta} \tilde H_{\A,n+k}(\hat{\theta}_{\A,n})-\pa_{\theta} \tilde H_{\A,n+k}(\theta_0)+\Big(1+\frac{k}{n}\Big)\pa_{\theta} \tilde H_{\A,n}(\theta_0)\Big\|=o_P(1).
\end{eqnarray*}
\end{lm}
\begin{proof}
The proof is essentially the same as the proof of Lemma \ref{lm4}. By Taylor's theorem, we have the same expansion for $\pa_\T \tilde H_{\A,n+k}(\hat \T_{\A,n})$ as in (\ref{TL}). By arguing analogously, one can have
\begin{eqnarray*}
&&\hspace{-0cm}\frac{1}{\sqrt{n} \Big(1+\frac{k}{n}\Big)}\Big\|\pa_{\theta} \tilde H_{\A,n+k}(\hat{\theta}_{\A,n})-\pa_{\theta}\tilde H_{\A,n+k}(\theta_0)+\Big(1+\frac{k}{n}\Big)\pa_{\theta} \tilde H_{\A,n}(\theta_0)\Big\|\\
&&\hspace{-0cm}\leq
\Big\|\big( \mathcal{J}_\A-\frac{1}{n+k}\paa \tilde H_{\A,n+k}(\theta^*_{n,k})\big)\mathcal{J}_\A^{-1}\frac{1}{\sqrt{n}}\pa_{\theta} \tilde H_{\A,n}(\theta_0)\Big\|\\
&&\hspace{0.5cm}
+\frac{1}{n+k}\big\|\paa \tilde H_{\A,n+k}(\theta^*_{n,k})\mathcal{J}_\A^{-1}
(\tilde B_{\A,n}-\mathcal{J}_\A)\sqrt{n}(\hat{\theta}_{\A,n}-\theta_0)\big\|\\
&&\hspace{-0cm}:=I_{n,k} +II_{n,k},
\end{eqnarray*}
where $\T^*_{n,k}$ is an intermediate point between $\T_0$ and $\hat \T_{\A,n}$. By condition {\bf S2}, we have
\begin{eqnarray}\label{Approx1}
\frac{1}{\sqrt{n}}\big\|\pa_\T H_{\A,n}(\T_0)-\pa_\T \tilde H_{\A,n}(\T_0)\big\| \leq
\frac{1}{\sqrt{n}}\sum_{t=1}^n\big\|\pa_{\T} l_\A(X_t;\T_0)-\pa_{\T} \tilde l_\A(X_t;\T_0)\big\|=o(1)\quad a.s.
\end{eqnarray}
Using condition {\bf S1} and the FCLT for martingale differences, we also have that for each $T>0$,
  \begin{eqnarray}\label{FCLT2}
\frac{1}{\sqrt{n}} \pa_\T
H_{\alpha,[ns]}(\theta_0)=\frac{1}{\sqrt{n}}  \sum_{t=1}^{[ns]}  \pa_\T l_\A (X_t; \T_0)\stackrel {w}{\longrightarrow} \mathcal{I}_\A^{1/2}W_d(s)~~in~~\mathbb {D}([0,T],\mathbb {R}^d ),
\end{eqnarray}
which together with (\ref{Approx1}) and Slutsky's theorem yields $\frac{1}{\sqrt{n}}\pa_\T \tilde H_{\A,n}(\T_0)=O_P(1)$.
 Combining this and (\ref{C1}), one can see $\sup_{k\geq1} I_{n,k}=o_P(1)$. Also, since  $\sqrt{n}(\hat\T_{\A,n}-\T_0)=O_P(1)$
 by assumption {\bf M4}, it follows from (\ref{C2}) that  $\sup_{k\geq1} II_{n,k}=o_P(1)$. This completes the proof.
\end{proof}

\begin{lm}\label{lm.I2}Suppose that condition {\bf S2} holds. Then, under  $H_0$,
\begin{eqnarray*}
\sup_{k\geq1}\frac{1}{\sqrt{n}\Big(1+\frac{k}{n}\Big)}\Big\| \pa_{\theta}\tilde{H}_{\A,n+k}(\theta_0)-\Big(1+\frac{k}{n}\Big)\pa_{\theta}\tilde{H}_{\A,n}(\theta_0)
-\pa_{\theta}{H}_{\A,n+k}(\theta_0)+\Big(1+\frac{k}{n}\Big)\pa_{\theta}{H}_{\A,n}(\theta_0)\Big\|=o(1)\quad a.s.
\end{eqnarray*}
\end{lm}
\begin{proof}
Since
\begin{eqnarray*}
&&\frac{1}{\sqrt{n}\Big(1+\frac{k}{n}\Big)}\Big\| \pa_{\theta}\tilde{H}_{\A,n+k}(\theta_0)-\Big(1+\frac{k}{n}\Big)\pa_{\theta}\tilde{H}_{\A,n}(\theta_0)
-\pa_{\theta}{H}_{\A,n+k}(\theta_0)+\Big(1+\frac{k}{n}\Big)\pa_{\theta}{H}_{\A,n}(\theta_0)\Big\|\\
&& \leq
 \frac{\sqrt{n}}{n+k}\sum_{t=1}^{n+k}\big\| \pa_{\theta}\,\tilde l_\A(X_t;\theta_0)-\pa_{\theta}\,l_\A(X_t;\theta_0)\big\|
+\frac{1}{\sqrt{n}}\sum_{t=1}^n \big\| \pa_{\theta}\,\tilde l_\A(X_t;\theta_0)-\pa_{\theta}\,l_\A(X_t;\theta_0)\big\|\\
&& \leq
\sup_{k\geq1}\frac{1}{\sqrt{n+k}}\sum_{t=1}^{n+k}\big\| \pa_{\theta}\,\tilde l_\A(X_t;\theta_0)-\pa_{\theta}\,l_\A(X_t;\theta_0)\big\|
+\frac{1}{\sqrt{n}}\sum_{t=1}^n \big\| \pa_{\theta}\,\tilde l_\A(X_t;\theta_0)-\pa_{\theta}\,l_\A(X_t;\theta_0)\big\|,
\end{eqnarray*}
the lemma is asserted  by condition {\bf S2} and Lemma \ref{lm1}.
\end{proof}

\begin{lm}\label{lm.I3}Suppose that condition {\bf S1} and assumption {\bf B} hold. Then, under  $H_0$,
\begin{eqnarray*}
\sup_{k\geq1}  \frac{\Big\|\pa_{\theta}{H}_{\A,n+k}(\theta_0)-\Big(1+\frac{k}{n}\Big)\pa_{\theta}{H}_{\A,n}(\theta_0)\Big\|}{\sqrt{n}\Big(1+\frac{k}{n}\Big)b\Big(\frac{k}{n}\Big)}
\stackrel{d}{\longrightarrow} \sup_{s>0} \frac{ \big\|\mathcal{I}_\A^{\frac{1}{2}}\big( W_d(1+s)-(1+s)W_d(1)\big)\big\|}{(1+s)b(s)}.
\end{eqnarray*}
\end{lm}
\begin{proof}
In exactly the same fashion as in Lemma \ref{lm5}, one can verify the lemma if it can be shown that the ones corresponding to(\ref{R1}) and (\ref{R3}) also hold
 for $\{H_{\A,n}(\T_0)\}$ defined in Section \ref{Sec:3}. Here, we note that (\ref{R2}) always holds under assumption {\bf B}.
   (\ref{R1}) follows from (\ref{FCLT2}) and the continuous mapping theorem. Using the Hájek-Rényi-Chow
  inequality with condition {\bf S1}, one can obtain the same result as (\ref{R3}). Thus, the lemma is asserted.
\end{proof}

\noindent{\bf Proof of Theorem \ref{thm4}}\\
Combining Lemmas \ref{lm.I1} - \ref{lm.I3}, we obtain the same result as (\ref{main1}). That is, under $H_0$,
\begin{eqnarray*}
\sup_{k\geq1}  \frac{\big\|\hat{\mathcal{I}}_{\A,n}^{-\frac{1}{2}}\,\pa_{\theta}\tilde{H}_{\A,n+k}(\hat{\theta}_n)\big\|}{\sqrt{n}\Big(1+\frac{k}{n}\Big)b\Big(\frac{k}{n}\Big)}
\stackrel{d}{\longrightarrow} \sup_{t>0} \frac{ \big\| W_d(1+t)-(1+t)W_d(1)\big\|}{(1+t)b(t)}.
\end{eqnarray*}
Hence, the theorem is established by using the same arguments used in the proof of Theorem \ref{thm2}.
\hfill{$\Box$}\vspace{0.2cm}\\

\noindent{\bf Proof of Theorem \ref{thm5}}\\
The theorem can be proved analogously to Theorem \ref{thm3}.
Let $\{k_n\}$ be an increasing sequence of positive integers satisfying $k_n/\sqrt{n}\rightarrow\infty$ and note that
\begin{eqnarray*}
\pa_{\theta}\tilde H_{\A,n+k_n}(\hat{\theta}_{\A,n})
&=&\sum_{t=1}^{n+k^*}\pa_{\theta}\,\tilde l_\A(X^{\T_0}_t;\hat\T_{\A,n})+\sum_{t>n+k^*}^{n+k_n}\pa_{\theta}\,\tilde l_\A(X^{\T_1}_t;\hat\T_{\A,n}).
\end{eqnarray*}
Since Lemmas \ref{lm.I1}-\ref{lm.I3} hold for $\{X_t^{\T_0} | t\geq1\}$, one can see that
\begin{eqnarray}\label{TS.AH1}
\frac{1}{\sqrt{n}\Big(1+\frac{k_n}{n}\Big)}\Big\|\sum_{t=1}^{n+k^*}\pa_\T\,\tilde l_\A(X^{\T_0}_t;\hat\T_{\A,n})\Big\|
=\frac{1}{\sqrt{n}\Big(1+\frac{k_n}{n}\Big)}O_P\Big(\sqrt{n}\Big(1+\frac{k^*}{n}\Big)\Big)=O_P(1).
\end{eqnarray}
The following uniform strong convergence can be also shown by Lemma \ref{unif} with condition {\bf S5}:
\begin{eqnarray*}
\sup_{\T\in N_3(\T_0)}\Big\| \frac{1}{k_n-k^*}\sum_{t>n+k^*}^{n+k_n}\pa_\T\,l_\A(X^{\T_1}_t;\T) - \E\, \pa_\T\,l_\A(X^{\T_1}_t;\T)\Big\| =o(1)\quad a.s.
\end{eqnarray*}
Then, due to the strong consistency of $\hat\T_{\A,n}$ and the continuity of $\E\, \pa_\T\,l_\A(X_t;\T)$, we have
\begin{eqnarray}\label{TS.AH20}
\Big\| \frac{1}{k_n-k^*}\sum_{t>n+k^*}^{n+k_n}\pa_\T\,l_\A(X^{\T_1}_t;\hat\T_{\A,n}) \Big\| \stackrel{a.s.}{\longrightarrow}\Big\|\E\, \pa_\T\,l_\A(X^{\T_1}_t;\T_0)\Big\| .
\end{eqnarray}
and thus by condition {\bf S6},
\begin{eqnarray*}
\frac{1}{k_n-k^*}\Big\| \sum_{t>n+k^*}^{n+k_n}\pa_\T\,\tilde l_\A(X^{\T_1}_t;\hat \T_{\A,n}) \Big\| \stackrel{a.s.}{\longrightarrow} E \big\| \pa_\T\,l_\A(X^{\T_1}_t;\T_0)\big\|.
\end{eqnarray*}
Hence, due to condition {\bf S7}, we have that
\begin{eqnarray}\label{TS.AH2}
\frac{1}{\sqrt{n}\Big(1+\frac{k_n}{n}\Big)}\Big\|\sum_{t>n+k^*}^{n+k_n}\pa_\T\,\tilde l_\A(X^{\T_1}_t;\hat\T_{\A,n})\Big\|
=\sqrt{n}\,\frac{k_n-k^*}{n+k_n}\big( E \big\| \pa_\T\,l_\A(X^{\T_1}_t;\T_0)\big\|+o(1)\big)\stackrel{a.s.}{\longrightarrow}\infty.
\end{eqnarray}
Therefore, we have by (\ref{TS.AH1}) and (\ref{TS.AH2}) that
\begin{eqnarray*}
\frac{\big\| \pa_\T \tilde H_{\A,n+k_n}(\hat{\theta}_{\A,n})\big\|}{\sqrt{n}\Big(1+\frac{k_n}{n}\Big)}
&\geq&
 \frac{1}{\sqrt{n}\Big(1+\frac{k_n}{n}\Big)}\Bigg\{\Big\|\sum_{t>n+k^*}^{n+k_n}\pa_\T\,\tilde l_\A(X^{\T_1}_t;\hat\T_{\A,n})\Big\|
-\Big\|\sum_{t=1}^{n+k^*}\pa_\T\,\tilde l_\A(X^{\T_0}_t;\hat\T_{\A,n})\Big\|\Bigg\}\\
\hspace{-6cm}&\stackrel{P}{\longrightarrow}& \infty,
\end{eqnarray*}
which establishes the theorem.
\hfill{$\Box$}
\begin{lm}\label{root.n.con}Suppose that assumptions {\bf M1}-{\bf M3} and conditions {\bf S1}-{\bf S4} hold.
If $\hat{\theta}_{\A,n}$ converges almost surely to $\theta_0$, then $\hat\T_{\A,n}$ is $\sqrt{n}$-consistency. More exactly,
\[ \sqrt{n}(\hat\T_{\A,n}-\T_0) \stackrel{d}{\longrightarrow} N\big(0, \mathcal{J}_\A^{-1}\mathcal{I}_\A\mathcal{J}_\A^{-1} \big).\]
\end{lm}
\begin{proof} Recall that $\frac{1}{\sqrt{n}}\pa_\T \tilde H_{\A,n}(\T_0)=O_P(1)$. Using (\ref{C1}) with $m=0$, one can also see that
$\tilde B_{\A,n}-\mathcal{J}_\A=o_p(1)$. Further, as in (\ref{root.hat}), it can be written that
\begin{eqnarray*}
\sqrt{n} (\hat{\theta}_{\A,n}-\theta_0)
&=& -\mathcal{J}_\A^{-1}\frac{1}{\sqrt{n}}\pa_\T \tilde H_{\A,n}(\theta_0)-
\mathcal{J}_\A^{-1}(\tilde B_{\A,n}-\mathcal{J}_\A)\sqrt{n}(\hat{\theta}_{\A,n}-\theta_0)\\
&=&O_P(1)+o_P(1)\sqrt{n} (\hat{\theta}_{\A,n}-\theta_0),
\end{eqnarray*}
implying that $\sqrt{n}(\hat\T_{\A,n}-\T_0)=O_P(1)$ (cf. Theorem 5.21 in \cite{van:1998}). Thus, we have by (\ref{Approx1}) that
\begin{eqnarray*}
\sqrt{n} (\hat{\theta}_{\A,n}-\theta_0)
= -\mathcal{J}_\A^{-1}\frac{1}{\sqrt{n}}\pa_\T H_{\A,n}(\theta_0)+o_P(1),
\end{eqnarray*}
which together with (\ref{FCLT2}) yields the lemma.
\end{proof}

\begin{lm}\label{inverse}
Suppose that assumptions {\bf G1}-{\bf G4} in Theorem \ref{thm6} hold for $\T_0$.  Then, $\mathcal{I}_\A$ and $\mathcal{J}_\A$ are positive definite.
\end{lm}
\begin{proof}
Using the equations (15) and (16) in Lemma 2 in \cite{lee:song:2009}, one can shows that
\[
 \mathcal{I}_\A=k(\A) \E \left[{\Big(\frac{1}{\sigma_t^2(\theta_0)}\Big)}^{\alpha+2}\,\pa_\T \sigma_t^2(\theta_0) \pa_{\T'}
\sigma_t^2(\theta_0)  \right] \]
\mbox{and}
\[ \mathcal{J}_\A=g(\A)\E \left[{\Big(\frac{1}{\sigma_t^2(\theta_0)}\Big)}^{\frac{\alpha}{2}+2}\,\pa_\T\sigma_t^2(\theta_0) \pa_{\T'}\sigma_t^2(\theta_0) \right],
\]
where
\begin{eqnarray*}
k(\alpha)=\frac{{(1+\alpha)}^2(1+2\alpha^2)}{2{(1+2\alpha)}^{2/5}}-\frac{\alpha^2}{4(1+\alpha)},\quad
g(\alpha)=\frac{\alpha^2+2\A+2}{4(1+\alpha)^{3/2}}.
\end{eqnarray*}
Noting first that  $k(\A)>0$ and $g(\A)>0$  for $\A\geq0$, we can see that for all $\lambda \in \mathbb{R}^{p+q+1}$,  $\lambda'\mathcal{I}_\A\lambda\geq0$ and $\lambda'\mathcal{J}_\A\lambda\geq0$.
According to the arguments in \cite{francq:zakoian:2004}, page 621, it holds that $\lambda' \pa_\T \sigma^2(\T_0) \stackrel{a.s.}{=}0$
implies $\lambda=0$.  Using this result,  it can be shown that $\lambda'\mathcal{I}_\A\lambda$ and   $\lambda'\mathcal{J}_\A\lambda$
are equal to zero  only for  $\lambda=0$, which asserts the lemma.
\end{proof}

\begin{lm}\label{S56}
Suppose that assumptions {\bf G1}-{\bf G4} in Theorem \ref{thm6} hold for $\T_1$.  If $\E [(X_t^{\T_1})^{4+\epsilon}] <\infty$ for some $\epsilon>0$, then for any subset $\widetilde \Theta$ of $\Theta$ satisfying $\min_{1\leq i \leq p}\inf_{\T\in\widetilde{\Theta}} \A_i >0$ and  $\min_{1\leq j \leq q}\inf_{\T\in\widetilde{\Theta}}\B_j>0$,
\begin{eqnarray}\label{S5}
\E  \sup_{\T \in \widetilde \Theta} \big\| \pa l_\A (X^{\T_1}_t;\T)\big\|  <\infty
\end{eqnarray}
and
\begin{eqnarray}\label{S6}
\sum_{t>n+k^*}^{n+k_n} \sup_{\T\in \widetilde \Theta} \big\| \pa_{\theta}\,l_\A(X^{\T_1}_t;\theta)-\pa_{\theta}\,\tilde l_\A(X^{\T_1}_t;\theta)\big\|= O(1)\quad a.s.,
\end{eqnarray}
where $\{k_n\}$ is any increasing sequence of positive integers.
\end{lm}
\begin{proof} In this proof, $K$ and $\rho$ are used to denote  generic constants taking different values $K>0$ and $\rho \in (0,1)$. For notational simplicity, we denote $X_t^{\T_1}$ by just $X_t$ without confusion.

Following the arguments in page 619 in \cite{francq:zakoian:2004}, it can be seen that
\begin{eqnarray}
&&
\Big|\frac{1}{\sigma_t^2}\frac{\pa\sigma_t^2}{\pa \W}\Big| \leq \frac{1}{\W}\sum_{k=0}^\infty \rho^k,\quad \Big|\frac{1}{\sigma_t^2}\frac{\pa\sigma_t^2}{\pa \A_i}\Big| \leq \frac{1}{\A_i},\label{BD1}
\end{eqnarray}
and for any $s\in(0,1)$,
\begin{eqnarray}
\Big|\frac{1}{\sigma_t^2}\frac{\pa \sigma_t^2}{\pa \B_j}\Big| \leq \frac{1}{\B_j \W^s} \sum_{k=0}^\infty k \rho^{ks} \big(\W^s+\sum_{i=1}^p\A_i^s X_{t-k-i}^{2s}\big).\label{BD2}
\end{eqnarray}
Note also that by Lemma 2.3 in \cite{berkes.et.al:2003}, for any $d\in\mathbb{N}$, there exists $s\in(0,1)$ such that $\E X_t^{2sd}<\infty$.
Hence, using (\ref{BD1}), (\ref{BD2}) and the Minkowski inequality, it can be shown that
\begin{eqnarray}\label{MC1}
\E \sup_{\T\in \widetilde \Theta}\Big|\frac{1}{\sigma_t^2} \pa_{\T_i} \sigma_t^2\Big|^d < \infty\quad \mbox{for all}\ d\in\mathbb{N},
\end{eqnarray}
where $\T_i$ denotes $i$-th element in $\T$.

By simple algebra, it can be written that
\begin{eqnarray*}\label{pa.l}
\pa_\T l_\A(X_t;\T)&=&h_{\A}(\sigma_t^2){\left(\frac{1}{\sigma_t^2}\right)}^{\frac{\A}{2}+1}\pa_\T\sigma_t^2,
\end{eqnarray*}
where
\begin{eqnarray*}
h_{\A}(x)=-\frac{\A}{2\sqrt{1+\A}}+\frac{1+\A}{2}\left(1-\frac{X_t^2}{x}\right)\exp\left(-\frac{\A}{2}\frac{X_t^2}{x}\right).
\end{eqnarray*}
Since $h_\A(\sigma_t^2) \leq K(1+X_t^2)$, we have
\[  \sup_{\T \in\widetilde \Theta} \big|\pa_{\T_i} l_\A(X_t;\T)\big| \leq K(1+X_t^2) \sup_{\T\in \widetilde \Theta}\Big|\frac{1}{\sigma_t^2} \pa_{\T_i} \sigma_t^2\Big|  \]
and thus (\ref{S5}) follows from the Cauchy-Schwarz inequality with the moment condition $\E X_t^{4+\ep} <\infty$ and (\ref{MC1}).

To show (\ref{S6}), we further use the following technical results in \cite{francq:zakoian:2004}:
\begin{eqnarray}
\sup_{\T\in\Theta}\big\{|\,\sigma_t^2-{\tilde{\sigma}}_t^2| \vee \|\pa_\T\sigma_t^2-\pa_\T \tilde{\sigma}_t^2\|\vee \|\paa  \sigma_t^2-\paa \tilde{\sigma}_t^2\|\big\} \leq K\rho^t \ a.s.\quad \text{for all } t\geq1, \label{FZ1}
\end{eqnarray}
Using (\ref{FZ1}), one can readily show that
\begin{eqnarray*}
|h_{\A}(\tilde{\sigma}_t^2)| \leq  K\big(1+X_t^2\big),\quad \big|h_{\A}(\sigma_t^2)-h_{\A}({\tilde{\sigma}}_t^2)\big| \leq K \big(X_t^2+X_t^4\big)\rho^t,\quad
\bigg|\left(\frac{1}{\sigma_t^2}\right)^{\frac{\A}{2}+1}-\left(\frac{1}{\tilde{\sigma}_t^2}
\right)^{\frac{\A}{2}+1}\bigg| \leq K\rho^t,\label{LS.3}
\end{eqnarray*}
where the last two inequalities can be derived by using the mean value theorem.  Thus, we have
\begin{eqnarray}\label{LS.4}
&&\big|\pa_{\T_i} l_\A(X_t;\T)-\pa_{\T_i} \tilde l_\A(X_t;\T)\big| \nonumber \\
&&=\bigg|\big(h_{\A}(\sigma_t^2)-h_{\A}({\tilde{\sigma}}_t^2)\big){\left(\frac{1}{\sigma_t^2}
\right)}^{\frac{\A}{2}+1}\pa_{\T_i}\sigma_t^2
+h_{\A}({\tilde{\sigma}}_t^2)\bigg\{\left(\frac{1}{\sigma_t^2}\right)^{\frac{\A}{2}+1}-\left(\frac{1}{\tilde{\sigma}_t^2}
\right)^{\frac{\A}{2}+1}\bigg\} \pa_{\T_i}\sigma_t^2\nonumber\\
&&
\quad+h_{\A}(\tilde{\sigma}_t^2)\left(\frac{1}{\tilde{\sigma}_t^2}
\right)^{\frac{\A}{2}+1}\big(\pa_{\T_i}\sigma_t^2-\pa_{\T_i}\tilde{\sigma}_t^2\big)\bigg|\nonumber\\
&&\leq K\big( X_t^2+X_t^4\big)\Big|\frac{1}{\sigma_t^2}\pa_{\T_i} \sigma_t^2\Big|\rho^t
+ K\big( 1+X_t^2\big)\Big|\frac{1}{\sigma_t^2}\pa_{\T_i} \sigma_t^2\Big|\rho^t
+ K\big( 1+X_t^2\big)\rho^t  \nonumber \\
&&\leq K\big( 1+X_t^2+X_t^4\big)\Big(1+\Big|\frac{1}{\sigma_t^2}\pa_{\T_i} \sigma_t^2\Big|\Big)\rho^t :=K P_{t,i}(\T)\rho^t.
\end{eqnarray}
Here, it should be emphasized that the superscript $t$ in $\rho^t$ denotes the time difference from the last initial value.
Observe  that by using H\"{o}lder's inequality, Minkowski's inequality, and (\ref{MC1}),  we also have
\begin{eqnarray*}
\E \sup_{\T \in \widetilde \Theta }P_{t,i}(\T)  &\leq & \big\| 1+X_t^2+X_t^4\big\|_{d_1}\, \Big\| 1+\sup_{\T \in \widetilde \Theta }\Big|\frac{1}{\sigma_t^2}\pa_{\T_i} \sigma_t^2\Big|\Big\|_{d_2}\\
&\leq & \big(1+\|X_t^2\|_{d_1}+\|X_t^4\|_{d_1} \big) \Big( 1+\Big\|\sup_{\T \in \widetilde \Theta }\Big|\frac{1}{\sigma_t^2}\pa_{\T_i} \sigma_t^2\Big|\Big\|_{d_2}\Big) <\infty,
\end{eqnarray*}
where $d_1=1+\epsilon/4$ and $d_2=1+4/\epsilon$.  Since $\tilde \sigma_t^2(\T)$ in $\pa_\T \tilde l_\A(X_t;\T)$ is indeed the one recursively  obtained with the initial values $\{X_t^{\T_0}, \tilde \sigma_t^2(\hat\T_{\A,n}) | t\leq n+k^*\}$, we can see from (\ref{LS.4})  that  $\big|\pa_{\T_i} l_\A(X_t;\T)-\pa_{\T_i} \tilde l_\A(X_t;\T)\big| \leq K P_{t,i}(\T) \rho^{t-(n+k^*)}$ for $t>n+k^*$, and thus,
\begin{eqnarray*}
\E\sum_{t>n+k^*}^{n+k_n} \sup_{\T\in \widetilde \Theta} \big| \pa_{\T_i}\,l_\A(X^{\T_1}_t;\theta)-\pa_{\T_i}\,\tilde l_\A(X^{\T_1}_t;\theta)\big| &\leq& K \sum_{t>n+k^*}^{n+k_n} \E\,\sup_{\T\in\widetilde\Theta} P_{t,i}(\T)\rho^{t-(n+k^*)}\\
&\leq&K(1-\rho^{k_n-k^*}).
\end{eqnarray*}
Therefore, by the dominate convergence theorem, we have
\begin{eqnarray*}
\E\lim_{n\rightarrow\infty}\sum_{t>n+k^*}^{n+k_n} \sup_{\T\in \widetilde \Theta} \big| \pa_{\T_i}\,l_\A(X^{\T_1}_t;\theta)-\pa_{\T_i}\,\tilde l_\A(X^{\T_1}_t;\theta)\big| <\infty,
\end{eqnarray*}
which asserts (\ref{S6}).
\end{proof}

\noindent{\bf Acknowledgements}\\
The author  would like to thank the associate editor for carefully examining the paper and providing valuable comments. This research was supported by Basic Science Research Program through the National Research Foundation of Korea (NRF) funded by the Ministry of Education (NRF-2019R1I1A3A01056924).

\bibliography{jun}

\end{document}